\documentclass[journal,draftclsnofoot,onecolumn]{IEEEtran}
\ifCLASSINFOpdf
   \usepackage[pdftex]{graphicx}
   \DeclareGraphicsExtensions{.pdf,.jpeg,.png,.eps}
\else
   \usepackage[dvips]{graphicx}
   \DeclareGraphicsExtensions{.pdf,.jpeg,.eps}
\fi
\usepackage[cmex10]{amsmath}
\usepackage{amssymb}
\usepackage{amsthm}

\newtheorem{defi}{Definition}
\newtheorem{theo}{Theorem}
\newtheorem{lem}{Lemma}

\usepackage{pgf}
\usepackage{algorithmic}
\usepackage{array}
\usepackage{url}
\begin{document}
\title{Finite Blocklength Analysis of Energy Harvesting Channels}
%
%
% author names and IEEE memberships
% note positions of commas and nonbreaking spaces ( ~ ) LaTeX will not break
% a structure at a ~ so this keeps an author's name from being broken across
% two lines.
% use \thanks{} to gain access to the first footnote area
% a separate \thanks must be used for each paragraph as LaTeX2e's \thanks
% was not built to handle multiple paragraphs
%

\author{K~Gautam~Shenoy,~\IEEEmembership{Member,~IEEE,}
        %John~Doe,~\IEEEmembership{Fellow,~OSA,}
        and~Vinod~Sharma,~\IEEEmembership{Senior Member,~IEEE}% <-this % stops a space
\thanks{Part of this paper has been published in ISIT 2016 \cite{Gau1}. K Gautam Shenoy and Vinod Sharma are with Electrical Communications Engineering Dept., Indian Institute of Science, Bangalore.}% <-this % stops a space
%\thanks{J. Doe and J. Doe are with Anonymous University.}% <-this % stops a space
%\thanks{Manuscript received April 19, 2005; revised August 26, 2015.}
}

\maketitle

% As a general rule, do not put math, special symbols or citations
% in the abstract or keywords.
\begin{abstract}
We consider Additive White Gaussian Noise channels and Discrete Memoryless channels when the transmitter harvests energy from the environment. These can model wireless sensor networks as well as Internet of Things. By providing a unifying framework that works for any energy harvesting channel, we study these channels assuming an infinite energy buffer and provide the corresponding achievability and converse bounds on the channel capacity in the finite blocklength regime. We additionally provide moderate deviation asymptotic bounds as well.
\end{abstract}

% Note that keywords are not normally used for peerreview papers.

% For peer review papers, you can put extra information on the cover
% page as needed:
% \ifCLASSOPTIONpeerreview
% \begin{center} \bfseries EDICS Category: 3-BBND \end{center}
% \fi
%
% For peerreview papers, this IEEEtran command inserts a page break and
% creates the second title. It will be ignored for other modes.
\IEEEpeerreviewmaketitle

\section{Introduction}
In the information theoretic analysis of channels, channel capacity is the maximum rate at which a source can transmit messages to the receiver subject to an arbitrarily small probability of error. However, channel capacity can be achieved arbitrarily closely by using very large blocklength codes. In practice, we are restricted by blocklength and as a result, we would like to study the backoff from capacity as well as the variation in maximal code size as a function of blocklength. For a fixed probability of error, the study of achievable rates in the finite blocklength regime is also known as a second order analysis in literature.

Like channel capacity, a finite blocklength characterization consists of two parts, namely the achievability and the converse bound on the maximal code size (number of messages) $M$. Given the probability of error, the achievability part usually deals with the existence of a code using, for instance, random coding arguments or manipulating general achievability bounds and showing that the bound can be achieved. The converse, on the other hand is an upper bound on the maximal code size which is to be satisfied by every feasible code. This paper focuses on developing both for the energy harvesting channels.

Energy harvesting (EH) channels and networks have gained considerable interest recently due to advances in wireless sensor networks and green communications (see \cite{kamal}, \cite{ku2016} and \cite{raza}). Transmitting symbols requires energy at the encoder end. Thus the study of the channel is done in tandem with the energy harvesting system. The energy harvesting section is modeled as a buffer or a rechargeable battery which stores incoming energy from some ambient source (e.g., solar energy from the sun). The energy buffer may be of finite or infinite length and the energy arrival process may be discrete or continuous. A problem of interest is to compare the performance of a channel with and without the energy harvesting system (e.g., whether we can quantify the impact on the channel capacity, finite blocklength capacity, etc.).

Finite blocklength analysis for discrete memoryless channels (DMC) was first carried out by Strassen \cite{Strs}. Hayashi \cite{Hayas} and Polyanskiy et.al. \cite{PPV1} provided non-asymptotic second order results for Additive White Gaussian Noise (AWGN) channels in addition to other channel types. \cite{PPV1,polyis} further provided the third order terms and developed a meta-converse, a converse result that recovered and improved upon known converses. Later, tighter results for various DMC's were studied by Tomamichel et al. in \cite{Tan2}. Non-asymptotic analysis of channels with state was carried out in \cite{Tan3}. Under the energy harvesting setup, assuming infinite buffer, the channel capacity for EH-AWGN channels was obtained in \cite{VSnRaj} and \cite{Uluk1}. The study of finite blocklength achievability for energy harvesting noiseless binary channels was carried out in \cite{Jyang}. Non-asymptotic achievability for EH-AWGN channels and EH-DMC's was developed in \cite{Tan1} where the second order term was shown to be $O(\sqrt{n\log n})$. Both Achievability and converse results for EH-AWGN channels were further refined recently in \cite{TanF2} which considered block i.i.d. energy arrivals. Finite blocklength analysis for fading channels under CSIT and CSIR was carried out in \cite{deeku1}.

In addition to finite blocklength analysis, we also give bounds on the moderate deviations coefficient for EH-AWGN channels and EH-DMC. In this analysis, we transmit at a rate less than capacity where the backoff goes to zero at a certain rate called the moderate deviation regime. In this regime, the probability of error will go to zero with increasing blocklength $n$. The goal is to characterize the moderate deviation error exponent. Moderate deviation analysis has been studied for memoryless channels by Altug and Wagner \cite{altug} as well as Polyanskiy and Verdu \cite{polymd}. In \cite{polymd}, the authors characterize the moderate deviation coefficient in terms of the channel dispersion. For DMCs with variable length feedback, the moderate deviation analysis was carried out by Truong and Tan \cite{truo}.

\subsection{Main Contributions}
% Revisit this part
In this paper, we provide a scheme that can be directly used to compute achievable rates for a wide class of energy harvesting channels. We focus on analyzing EH-AWGN and EH-DMC with infinite buffer. We assume a fixed maximum probability of error while analyzing both.The scheme improves over previously known bounds and also provides the achievability of $\sqrt{n}$. It is shown that a save and transmit scheme where the saving phase is $O(\sqrt{n})$ long is sufficient to allow for reliable communication in an energy harvesting set up. When compared with the non-energy harvesting case (but with average power constraint), we observe that the second order term is still $\Theta(\sqrt{n})$. Note that the coefficients of the second order term would not necessarily be same. 

Next we provide a finite blocklength converse for EH-AWGN channels. This is derived by modifying Polyanskiy et. al. meta converse \cite{PPV1} and specifically applying it to EH-AWGN channels.We are able to show that in both, the achievability and converse, the second order term is $O(\sqrt{n})$. This also gives us the strong converse for this channel for free as the first order term is unaffected by the probability of error term. Next, we analyze DMCs with energy harvesting and provide the finite blocklength achievability and converse bounds for them. Then, we provide moderate deviation lower and upper bounds for both types of channels. This is done by showing that the bounds on channel dispersion, obtained while proving the finite blocklength bounds, also bound the moderate deviations coefficient. Finally, we plot our bounds for certain parameters and provide suitable inferences.

Recently, there have been improvements to the results of EH-AWGN channels notably in the converse bound \cite{TanF2}. Our proof is an alternate proof of the same under maximal probability of error criterion and the framework we consider is useful in obtaining a converse for EH-DMC.

%The paper is organized as follows. We provide the notation and some basic results in Section \ref{prem}. We explain the details of the energy harvesting system for the AWGN channel in section \ref{ehaw}. Section \ref{FBAchi} deals with the finite blocklength achievability for an EH-AWGN channel and Section \ref{FBConv} provides the finite blocklength converse for an EH-AWGN channel. In section \ref{FBDMC}, we derive the achievability and converse bounds for an energy harvesting DMC. Finally, we conclude the paper.
\section{Preliminaries}\label{prem}
\subsection{Basic notation}
We shall use boldface letters (e.g. $\mathbf{x}$) to denote vectors (belonging to $\mathbb{R}^n$ for a specified $n\in \mathbb{N}$). When the length of a vector needs to be specified, we shall mention it as $\mathbf{x}^{k}=(x_1,x_2,\cdots,x_k)$. Similarly, $\mathbf{x}_i^j = (x_i,x_{i+1},\cdots,x_j)$. Lower case letters denote deterministic scalars or vectors whereas upper case letters denote random variables or random vectors respectively. We shall use $[M]$ to denote the set $\{1,2,\cdots,M\}$. We shall denote by $\mathcal{P}(\mathcal{X})$, the set of probability distributions on $\mathcal{X}$ (in cases the alphabet is clear, we simply use $\mathcal{P}$). The expectation operator will be denoted by $\mathbb{E}$ and if the distribution (say $P$) needs to be specified, then it shall be denoted as $\mathbb{E}_P$. We will occasionally use the Bachmann-Landau notation $O(.)$, $\Theta(.)$ etc. to denote appropriate orders.

\subsection{Channels, probability of error and capacity}\label{bdef2}
Given an input alphabet $\mathcal{X}$ and output alphabet $\mathcal{Y}$, a \textit{channel}, denoted by $W(y|x)$ or equivalently $P_{Y|X}$, is a conditional probability measure on $\mathcal{Y}$ given $x \in \mathcal{X}$. If the probability density function exists for the channel, we shall denote it by $f_{Y|X}$.

Given a probability distribution $P$ on $\mathcal{X}$ and a channel $W$, we define the output measure $PW$ as
\begin{equation}
PW(y) = \sum_{x\in \mathcal{X}} P(x)W(y|x).\notag
%\label{marg}
\end{equation}
There are two notions of probability of error which we will use. Given a code $\mathcal{C}$ with $M$ messages, let $U \in [M]$ be the random variable, uniformly distributed on $[M]$, denoting the message to be transmitted and $\hat{U} \in [M]$ the message that is decoded at the receiver. The \textit{maximal probability of error} (max p.o.e.) of the code $\mathcal{C}$ is
\begin{IEEEeqnarray}{rCl}
P_{e,max}(\mathcal{C}):= \max \limits_{1\leq m\leq M} Pr\left[\hat{U}\neq m | U=m \right].%\notag
\label{maxpoe}
\end{IEEEeqnarray}
Similarly the \textit{average probability of error} (avg p.o.e.)is defined as
\begin{IEEEeqnarray}{rCl}
P_{e,avg}(\mathcal{C}):= \frac{1}{M}\sum_{m=1}^M Pr\left[\hat{U}\neq m | U=m \right].\notag
%\label{avgpoe}
\end{IEEEeqnarray}
The channel capacity is the same in both cases. However, in the finite blocklength regime, the differences are in higher order terms resulting in an $O(\log n)$ difference \cite{PPV1}. In this paper, we will stick to the maximal probability of error criterion since it is advantageous while analysing the energy harvesting DMC results.

An $(n,M,\varepsilon)$ code is a code with $M$ codewords of codeword length $n$ and probability of error at most $\varepsilon$. We define
\begin{IEEEeqnarray}{rCl}
M^*(n,\varepsilon) := \max \{M: \mbox{There exists a }(n,M,\varepsilon)\mbox{ code}\}.\notag
%\label{bcode1}
\end{IEEEeqnarray}
Given a $(n,M,\varepsilon)$ code, we shall call $\frac{\log M}{n}$ as the \textit{rate} of the code. For $0<\varepsilon<1$, the \textit{$\varepsilon-$capacity} $C_{\varepsilon}$ is defined as
\begin{IEEEeqnarray}{rCl}
C_{\varepsilon} = \lim_{n \to \infty} \frac{\log M^*(n,\varepsilon)}{n}\notag
%\label{ecap}
\end{IEEEeqnarray}
and the \textit{capacity} of the channel is defined as
\begin{IEEEeqnarray}{rCl}
C = \lim_{\varepsilon \to 0} C_{\varepsilon}.\notag
%\label{cap}
\end{IEEEeqnarray}
Note that both limits exist. It is clear that $C_{\varepsilon} \geq C$. However for certain classes of channels like DMCs and standard AWGN channels with average power constraints, we have $C_{\varepsilon} = C$ for every $0<\varepsilon <1$. Then we say that the channel satisfies the \textit{strong converse}, which means that if we transmit at rates greater than capacity, the probability of error of the code tends to 1 as the blocklength $n$ tends to infinity.

\subsection{AWGN Channel}
Given $\mathbf{a} \in \mathbb{R}^n$ and a covariance matrix $\mathbf{K} \in \mathbb{R}^{n\times n}$, denote
\begin{IEEEeqnarray}{rCl}
\mathcal{N}(\mathbf{a};\mathbf{K}) := \frac{\exp\left\{-(\mathbf{x}-\mathbf{a})^T\mathbf{K}^{-1}(\mathbf{x}-\mathbf{a})\right\}}{(2\pi)^{n/2}(\det(\mathbf{K}))^{1/2}}\notag
%\label{Norm1}
\end{IEEEeqnarray}
as the multivariate normal distribution with mean $\mathbf{a}$ and covariance matrix $\mathbf{K}$ whose determinant is non-zero. An additive white Gaussian noise (AWGN) channel with noise variance $\sigma^2$ is given by
\begin{IEEEeqnarray}{rCl}
Y = x+Z \nonumber
\end{IEEEeqnarray}
where $x \in \mathbb{R}$ is the input to the channel and $Z \sim \mathcal{N}(0;\sigma^2)$. The n-dimensional version is obtained by applying the one dimensional version ($n=1$) case independently on each input $x_i$, $1\le i\le n$. The AWGN channel with average power constraint $S$ is an AWGN channel where the input $\mathbf{x}$ satisfies 
\begin{IEEEeqnarray}{rCl}
\|\mathbf{x}\|_2^2 \leq nS 
\label{pwrcon}
\end{IEEEeqnarray}
where for $p \geq 1$, $\|\mathbf{x}\|_p = \left(\sum \limits_{i=1}^n x_i^p\right)^{1/p}$ is the $p$th norm of $\mathbf{x}$.

The capacity of an AWGN channel (denoted by $C_G$) with average power constraint $P$ is given by
\begin{IEEEeqnarray}{rCl}
C_G := \frac{1}{2}\log_2\left(1+\frac{P}{\sigma^2}\right) \mbox{ bits per channel use}.\notag
%\label{Cap1}
\end{IEEEeqnarray}
%Moreover it is known that the strong converse holds for standard AWGN channels and so the $\varepsilon-$capacity is equal to $C_G$. This is also evident from the finite blocklength characterization of capacity (see
In \cite{Strs} and \cite{Hayas}, it was shown that for an AWGN channel with average power constraint $P$, the maximum code size $M^*(n,\varepsilon,P)$, for $n$ sufficiently large, satisfies
\begin{equation}
\log M^*(n,\varepsilon,P) = nC_G +\sqrt{nV_G}\Phi^{-1}(\varepsilon) + O(\log(n))\notag
%\label{FBgau}
\end{equation}
where the probability of error is at most $\varepsilon$, $V_G = \frac{P(P+2)}{2(P+1)^2}$ and $\Phi(x) = \int \limits_{-\infty}^x \frac{e^{-u^2/2}}{\sqrt{2\pi}}\,du$. 

\subsection{Discrete Memoryless Channels (DMC)}
A DMC is characterized by a finite input alphabet $\mathcal{X}$, finite output alphabet $\mathcal{Y}$ and the transition probabilities given by $W=P_{Y|X}$, which satisfies for every $n \geq 1$
\begin{equation}
P_{\mathbf{Y}|\mathbf{X}}(\mathbf{y}|\mathbf{x}) = \prod_{i=1}^n P_{Y|X}(y_i|x_i).\notag
%\label{dmc1}
\end{equation}
While the output is, in principle, allowed to depend on past outputs (which is known as a DMC with feedback), we only consider DMC's without feedback. The capacity $C_D$ of a DMC $W$ is given by Shannon's formula as
\begin{equation}
C_D = \sup_{P \in \mathcal{P}(\mathcal{X})} I(P;W) \triangleq \sup_{P \in \mathcal{P}(\mathcal{X})} \sum_{x \in \mathcal{X}}\sum_{y \in \mathcal{Y}}P(x)W(y|x)\log\left(\frac{W(y|x)}{PW(y)}\right),\notag
%\label{dmc2}
\end{equation}
where $I(P;W)$ is the mutual information (see \cite{Cover}) between the input and output of the channel. 

Now we define a few terms that will be required later.
\begin{defi}
Given a channel $W$ and an output distribution $Q$, the information density \cite{Tshan} of the channel is given by
\begin{equation}
i(x,y;Q) =\log \left(\frac{W(y|x)}{Q(y)}\right).\notag
%\label{idens}
\end{equation}
Often $Q = PW$ for some $P \in \mathcal{P}(\mathcal{X})$, in which case we shall denote the information density by $i_P(x,y)$.
\end{defi}
Observe that $I(P;W) = \sum \limits_{x,y} P(x)W(y|x)i_P(x,y)$. Similarly, the variance of information density is given by
\begin{IEEEeqnarray}{rCl}
V(P;W) := \left[\sum \limits_{x \in \mathcal{X}}\sum \limits_{y \in \mathcal{Y}}P(x)W(y|x) (i_P(x,y))^2\right] -(I(P;W))^2.\notag
%\label{dmcv2}
\end{IEEEeqnarray}

The finite blocklength result for DMC's with channel $W$, probability of error $0<\varepsilon<1$ and $V(P;W) >0$ for a capacity achieving distribution $P$ is given by (see \cite{Strs}-\cite{Tan2})
\begin{equation}
\log M^*(n,\varepsilon) = nC_{D} +\sqrt{nV_D}\Phi^{-1}(\varepsilon) + O(\log(n)),\notag
%\label{dmcf}
\end{equation}
where 
\begin{IEEEeqnarray}{rCl}
V_D =\left\{ \begin{array}{c c}
V_{\min} := \min \limits_{P \in \Pi} V(P;W), & \varepsilon \leq 1/2, \\
V_{\max} := \max \limits_{P \in \Pi} V(P;W), & \varepsilon > 1/2, 
\end{array} \right..\notag
%\label{dmcv}
\end{IEEEeqnarray}
and $\Pi=\{P\in \mathcal{P}(\mathcal{X}): I(P;W)=C_D\}$ is the set of capacity achieving distributions. 

\subsection{DMC with cost constraints}\label{dmcost}
Let $\Lambda:\mathcal{X} \to \mathbb{R}$ be a non-negative function which we will refer to as the energy function. The energy function simply returns the energy of the symbol $x$ which is a generalization of the standard energy function $\Lambda(x) = x^2$ for an AWGN channel. We further assume that the energy function is separable, i.e., given a vector $\mathbf{x} \in \mathcal{X}^n$,
\begin{equation}
\Lambda(\mathbf{x}) := \sum_{i=1}^n \Lambda(x_i).
\label{energ1}
\end{equation}

Define the constrained sets $\mathbb{F}_a$ and $\mathcal{F}_a$ for $a\geq0$ as follows
\begin{IEEEeqnarray}{rCl}
\mathbb{F}_a &=& \{\mathbf{x}\in \mathcal{X}^n: \, \Lambda(\mathbf{x}) \leq na\}, \label{conset0} 
\\
\mathcal{F}_a &=& \{ P\in \mathcal{P}: \, \mathbb{E}_P[\Lambda(X)] \leq na\}.
\label{conset}
\end{IEEEeqnarray}

In a DMC with cost constraints (see \cite{Csiz,Kost}), where the codewords are drawn from $\mathbb{F}_a$, the capacity changes to 
\begin{equation}
C_D(a)=\sup\limits_{P \in \mathcal{F}_a}I(P;W). 
\label{cost1}
\end{equation}

Moreover, the maximum achievable code size, for any $a>0$, denoted by $M^*(n,\varepsilon,a)$, under some regularity conditions (see \cite{Hayas} for the result and \cite{Kost} for some refinements), is given by
\begin{equation}
\log M^*(n,\varepsilon,a) = nC_D(a) + \sqrt{nV_D(a)}\Phi^{-1}(\varepsilon) + O(\log n) \notag
%\label{Fcost}
\end{equation}
where $V_D(a)$ is the channel dispersion (see \cite{Hayas}).

An energy harvesting DMC (EH-DMC), may be viewed as a generalization of a DMC with cost constraints and its finite blocklength analysis is reserved to Section \ref{FBDMC}.

% Write about 3 pages
%First let us provide the notation. 
%\subsection{Discrete Memoryless Channels and EH-DMC}
%A discrete memoryless channel(DMC) is described by an input alphabet $\mathcal{X}$, output alphabet $\mathcal{Y}$, both of which are finite, and a probability mass function (p.m.f) $W:\mathcal{X}\to\mathcal{Y}$. The p.m.f will be compactly denoted as $W(y|x)$ for every $x\in \mathcal{X}$ and $y \in \mathcal{Y}$. Note that we are considering DMCs without feedback. Given an input vector $\mathbf{x}=\mathbf{x}=\{x_1,x_2,\cdots,x_n\}$, the output vector $\mathbf{y}=y^n=\{y_1,y_2,\cdots,y_n\}$ is distributed according to 
%\begin{equation}
%W^n(\mathbf{y}|\mathbf{x}) := \prod_{i=1}^n W(y_i|x_i)
%\label{prod1}
%\end{equation}
\subsection{Energy Harvesting AWGN channel} \label{ehaw}
An energy harvesting system consists of an energy buffer which stores energy from various sources over a period of time. Energy is usually harvested from some ambient source, e.g., solar power. An EH-AWGN channel consists of an energy harvesting system at the transmitter end, followed by an AWGN channel as shown in Fig. \ref{Fi1}.
	\begin{figure}[ht]
	\centering
	\includegraphics[scale=0.6]{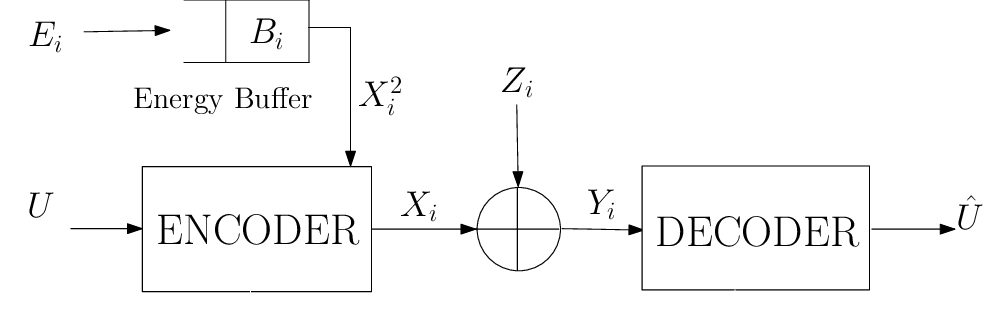}
	\caption{Block diagram of an AWGN energy harvesting system} \label{Fi1}
	\end{figure}
	To send symbol $x$ on the channel, we would require $x^2$ units of energy from the buffer and if sufficient energy is available, the transmission succeeds; otherwise an outage occurs. An outage event may be handled as an error event or a suitably truncated symbol may be transmitted. In this paper, as far as achievability is concerned, the outage will be treated as an error event. We assume that the energy buffer has infinite capacity and energy leakages do not occur. Additionally, the incoming energy process $\{E_i\}$ is assumed to be i.i.d., non negative with finite mean $\mathbb{E}[E_1]$ and finite variance $\sigma^2_E$.

The system works as follows. We first harvest energy $E_i$ in slot $i$, use it along with some energy in the buffer if needed to transmit the symbol and then store the remaining energy. Let $B_i$ be the energy in buffer at the $i$th transmission slot. Assume $B_0 = 0$. Then the energy in buffer, for $1\leq i \leq n$ evolves as
\begin{equation}
B_{i}=(B_{i-1}+E_i -X_i^2 )^+\notag
%\label{bufeq}
\end{equation}
where $(x)^+=\max\{x,0\}$.

For the ordinary AWGN channel with power constraint $S$, the sequences were supposed to satisfy (\ref{pwrcon}). The constraint for the energy harvesting AWGN channel, with input $\mathbf{x}$ and energy vector $\mathbf{e}$ is
\begin{equation}
\|\mathbf{x}^k\|_2^2 \leq \|\mathbf{e}^k\|_1 \quad 1\leq k \leq n.\notag
%\label{ehcons}
\end{equation}
which is another way of saying that there should, at every time instant, be enough energy to transmit the desired symbol. To ensure this, $\mathbf{x}$ is allowed to depend on $\mathbf{e}$. 

The capacity of an EH-AWGN channel (see \cite{VSnRaj} and \cite{Uluk1}) is 
\begin{equation}
C_{EG} = \frac{1}{2}\log\left(1+\frac{\mathbb{E}[E_1]}{\sigma^2}\right).
\label{capeq}
\end{equation}
Additionally, the strong converse was also shown to hold for this channel (see \cite{VSnRaj}). This would logically imply a converse of the form $\log M \leq nC_{EG} + o(n)$. However as we seek a refinement of this expression, we would need finer tools to extract a finite blocklength converse. In this regard, we will be using several results from \cite{PPV1}. For clarity, we use the notation of that paper.

We now state the following bounds on finite blocklength capacity for EH-AWGN channels.

\begin{theo}\label{TH1}
Consider a EH-AWGN channel, with AWGN variance $\sigma^2$, with energy harvesting process $\{E_i\}$, i.i.d. at the encoder, with mean $\mathbb{E}[E_1]$ and variance $\sigma_E^2 < \infty$. Given maximal probability of error $\varepsilon>0$,
\begin{enumerate}
	\item (Achievable bound)  For sufficiently large blocklength $n$, the maximum size of the code, $M^*(n,\varepsilon)$, satisfies
\begin{equation}
\log M^*(n, \varepsilon) \geq nC_{EG}+ \sqrt{n}\left[\sqrt{V_{EG}}\Phi^{-1}\left(\lambda\varepsilon\right)-K_{\varepsilon, \lambda}C_{EG}\right] - \frac{1}{2}\log n + O(1),
\label{ehawgnclb}
\end{equation}
where $C_{EG}$ is defined in \eqref{capeq}, $V_{EG} = \frac{\mathbb{E}[E_1]}{\mathbb{E}[E_1] + \sigma^2}\log^2_2(e)$, $K_{\varepsilon, \lambda} = \sqrt{\frac{4(2\mathbb{E}[E_1]^2 + \sigma_E^2 )}{(1-\lambda) \varepsilon \mathbb{E}[E_1]^2}}$ and the above holds for any $0<\lambda<1$.

\item (Converse Bound) Also,
\begin{equation}
\log M^*(n, \varepsilon) \leq nC_{EG} +\sqrt{nV_{EG2}}\Phi^{-1}(\varepsilon) + \frac{1}{2}\log n +O(1),
\label{ehawgncub}
\end{equation}
where $V_{EG2} = \frac{\mathbb{E}[E_1]^2 + \mathbb{E}[E_1^2] + 4\sigma^2\mathbb{E}[E_1]}{4(\mathbb{E}[E_1]+\sigma^2)^2}\log_2^2(e)$.
\end{enumerate}
\end{theo}
We defer the proof of the achievable bound to Section \ref{AWAch} and the converse bound to Section \ref{awgconv}. While the second order terms (coefficients of $\sqrt{n}$) do not match, we can conclude that $\log M^*(n,\varepsilon) = nC_{EG} + \Theta(\sqrt{n})$.
\subsection{Energy Harvesting DMC}\label{subdm}
An energy harvesting DMC is a DMC with an energy harvesting set up at the encoder. Let $\Lambda(.)$ be the energy function (see \eqref{energ1}) associated with this DMC. The model is the same as that of an EH-AWGN channel except for the following differences and assumptions.
\begin{enumerate}
	\item The AWGN channel is replaced with a DMC.
	\item Energy consumed by symbol $x_i$ is $\Lambda(x_i)$. Also there is a symbol $x_0$, with $\Lambda(x_0)=0$.
	\item We additionally assume that the DMCs are not exotic\footnotemark (see Appendix H of \cite{PPV1} and also see \cite{Tan2}). \footnotetext{A DMC is exotic if the maximum variance of it's information density i.e. $V_{\max}=0$ and for some input symbol $x_0$, $P(x_0) = 0$ for any capacity achieving distribution $P$ but $D(W(.|x_0)||Q_Y^*) = C$ and $V(W(.|x_0)||Q_Y^*) >0$.}
\end{enumerate}
The analysis for energy harvesting DMC's is, by and large, analogous to the analysis of EH-AWGN channels. However, using method of types (refer \cite{Cover,Csiz} for more information on types), we are able to improve the converse bound to resemble that of the original non-energy harvesting DMC. 

The capacity of an EH-DMC, where the energy harvesting process has mean $\mathbb{E}[E_1]$, is given by 
\begin{equation}
C_{ED} := \sup_{P\in \mathcal{F}_{\mathbb{E}[E_1]} } I(P;W)
\label{EDcap}
\end{equation} 
where $\mathcal{F}_a$ was defined in \eqref{conset}.

The following are the finite blocklength bounds on rate for EH-DMC, proved in this paper.

\begin{theo}
	Given $0<\varepsilon <1$, for maximal probability of error, consider an EH-DMC with HUS architecture, and the energy process $\{E_i\}$ i.i.d. with $E[E_1^2] < \infty$.
	\begin{enumerate}
		\item (Achievable bound) Given the input distribution $P_X \in \mathcal{F}_{\mathbb{E}[E_1]}$, the maximal size of the code $M^*(n,\varepsilon)$ with blocklength $n$ sufficiently large, satisfies
		\begin{equation}
		\log M^*(n,\varepsilon) \geq nI(P_X;W) - \sqrt{n}K_{\varepsilon,\lambda} I(P_X;W) + \sqrt{nV(P_X;W)}\Phi^{-1}\left(\lambda\varepsilon\right) -\frac{1}{2}\log n +O(1),
		\label{achdmc}
		\end{equation}
			for any $0 <\lambda < 1$. Here  $K_{\varepsilon,\lambda} = \frac{2\sqrt{Var(\Delta_1)}}{\mathbb{E}[E_1]\sqrt{(1-\lambda)\varepsilon}}$ and $\Delta_1 = E_1 -\Lambda(X_1)$. 
		
		\item (Converse Bound)
			Given $\eta > 0$, the maximal size of the code $M^*(n,\varepsilon)$ satisfies
			\begin{equation}
			\log M^*(n,\varepsilon) \leq nC_{ED} +\sqrt{n}C'(\mathbb{E}[E_1])D_{\varepsilon} +   \sqrt{nV^*_{\varepsilon}(\eta)}\left(\Phi^{-1}(\varepsilon) + \frac{K_{\varepsilon}\varepsilon}{4} \right) +O(\log n),
			\label{dmcfinq0}
			\end{equation}
			where $C'()$ is the derivative of the capacity cost function given in (\ref{cost1}) and $D_{\varepsilon}$, $K_{\varepsilon}$ and $V_{\varepsilon}^*(\eta)$ are functions of $\varepsilon$ independent of $n$.
			
	\end{enumerate}

\end{theo}
\subsection{Encoder and Decoder for Energy Harvesting Channels}
For traditional channels (AWGN, DMC etc.), the encoder and decoder have access to the codebook (random or otherwise) for the purposes of encoding and decoding respectively. In the energy harvesting setup, the encoder has access to the incoming energy values. Hence any codeword $c\in \mathcal{C}$, where $\mathcal{C}$ is the codebook, is an $n$ length vector  $c(m,\mathbf{e}^n)$, where $n$ is the block length, for message $m$ and energy vector $\mathbf{e}^n$. Due to causality requirements, the $i$th symbol of the codeword can only depend on $\mathbf{e}^i$. The decoder does not have access to these energy values and therefore does not have access to this energy dependent codebook. The decoder is, on the other hand, allowed to have access to an energy independent pre-codebook. Henceforth, in the context of energy harvesting channels, a codeword corresponding to message $m$ shall mean the mapping $m\to c(m,.)$. We note that the definitions of code size $M$, probability of error etc. as defined in section \ref{bdef2} remain unchanged. This is in the same spirit of the analysis of channels with state information available at encoder only.

We shall see in the achievability proofs that a codebook that is independent of energy values is created which is also available at the decoder. Then, at the encoder, the energy values are used to modify the codewords so as to meet the necessary constraints. This is one way of creating an encoder-decoder pair. This concept is similar to the one used in channels with state where the encoder has state information but not decoder \cite{elg}.

\section{Finite Blocklength Achievability bound for EH-AWGN}\label{AWAch}
This section deals with the proof of Theorem \ref{TH1}, part $1$). Let $0<\varepsilon<1$ be given. We shall construct a code for the EH-AWGN channel, which will have maximal probability of error not more than $\varepsilon$. Assume the buffer is empty at the beginning. This gives the worst case scenario, for if the buffer were nonempty at the start it could only help communications and therefore, our bound would still hold. The coding scheme we propose has two phases; namely a saving phase and a transmission phase. This is known in literature as the \textit{save and transmit scheme} (see \cite{Uluk1}). 
\subsection{Saving Phase}
In this phase, the transmitter transmits $0$, the symbol that uses zero energy, for a set number of slots. During this period, it allows the energy buffer to build up. The receiver is aware of the number of slots and chooses to ignore the output during those slots since they are not information bearing. The caveat is that slots are wasted, as far as information transfer is concerned, in gathering energy. To ensure that this scheme does not affect the coefficient of the first order term, it is required that the number of slots set for gathering energy scale at most as $o(n)$. 
%After this phase, the next phase is the transmission phase wherein the code book symbols are transmitted. Since we begin transmitting with a buffer that isn't empty, this will improve the chances of non occurence of energy outage. But the caveat is that slots are wasted in gathering energy.  

Fix $0<\lambda<1$ and consider $K_{\varepsilon,\lambda}$ as given in the statement of the theorem. Let $N_n$ represent the number of slots reserved for the saving phase. During this phase, the buffer fills up with energy and after $N_n$ time slots, we expect it to have crossed some threshold energy value which we will denote by $E_{0n}$. Let $N_n = K_{\varepsilon,\lambda}\sqrt{n}$, where and $E_{0n}=N_n\mathbb{E}[E_1]/2$. Let $\mathcal{E}_0$ denote the event that the system failed to gather $E_{0n}$ energy. We have
\begin{IEEEeqnarray}{rCl}
Pr(\mathcal{E}_0) &=& Pr\left[\sum_{i=1}^{N_n} E_i \leq E_{0n}\right] \notag\\
&=& Pr\left[\sum_{i=1}^{N_n} (E_i -\mathbb{E}[E_1]) \leq -E_{0n}\right] \notag\\
&\leq& Pr\left[\left| \sum_{i=1}^{N_n} (E_i -\mathbb{E}[E_1]) \right| \geq E_{0n}\right]\notag\\
&\leq& \frac{4\sigma_E^2}{K_{\varepsilon,\lambda}\mathbb{E}[E_1]^2\sqrt{n}}
\label{egout}
\end{IEEEeqnarray}
where in the last step, we used Chebyshev's inequality. The above bound ensures a decay of  $O(n^{-1/2})$ in the probability of error and hence can be made arbitrarily small for fixed $\varepsilon$.
\subsection{Transmit phase}
Let $n$ be the number of slots wherein we transmit symbols on the AWGN channel. We count channel uses from the $N_n +1$ instant onwards. Once we gather at least $E_{0n}$ energy, we must ensure, with high probability, that subsequent transmissions will not cause an outage. Let $\mathbf{v}^n$ be the input before checking the energy buffer. At transmission instant $i$, $1\leq i\leq n$, there are two cases, i.e.,
\begin{enumerate}
	\item There is sufficient energy in which case the input to the channel $x_i = v_i$.
	\item There is insufficient energy in which case we transmit $x_i=0$.
\end{enumerate}

 Let us denote the set of sequences $(\mathbf{v}^n, \mathbf{e}^n)$ that satisfy the above requirements by $\mathcal{A}_n$ where  
%satisfying the above requirements by $\mathcal{A}_n$ where
\begin{equation}
\mathcal{A}_n = \bigcap_{l=1}^n\{s_l \geq -E_{0n} \},
\label{cond1}
\end{equation}
and $s_l = \sum \limits_{k=1}^l e_k -v_k^2$. Note that the transmitted codeword satisfies the energy harvesting conditions, since $E_{0n}$ energy has already been harvested before the transmission started. Denote by $\mathcal{E}_1$ the event that the energy constraints are violated. Let $\{V_i\}$, $1\leq i\leq n$ be i.i.d. random variables (not necessarily Gaussian) with zero mean and variance $\mathbb{E}[E_1]$. Formally, 
\begin{IEEEeqnarray}{rCl}
Pr(\mathcal{E}_1) &=& Pr(\mathcal{A}_n^c) \notag \\
&=& Pr\left[\bigcup_{l=1}^n \left\{S_l \leq -E_{0n} \right\} \right]\notag\\
&\leq& Pr\left[\bigcup_{l=1}^n \left\{|S_l| \geq E_{0n} \right\} \right]\notag\\
&=& Pr\left[\max_{1\leq l\leq n} |S_l| \geq E_{0n}\right]
\label{ehviol}
\end{IEEEeqnarray}
and $S_l = \sum \limits_{k=1}^l E_k -V_k^2$. Now $S_l$ is a sum of i.i.d. random variables with zero mean and finite variance. We now invoke Kolmogorov's inequality (\cite{Athr}, Chapter 3) which is stated as follows.

\begin{lem}[Kolmogorov's Inequality]
Let $Z_i$ be independent zero mean random variables and $S_n = \sum_{i=1}^nZ_i$. If $\mathbb{E}[Z_i] = 0$ and $\mathbb{E}[Z_i^2] < \infty$ then for any $0 < a < \infty$
\begin{IEEEeqnarray}{C}
P\left(\max_{1\leq i\leq n}|S_i| \geq a\right) \leq \frac{\mathbb{E}[S_n^2]}{a^2}.\notag
%\label{Kolm}
\end{IEEEeqnarray}
\end{lem}

Hence we have,
\begin{IEEEeqnarray}{rCl}
Pr(\mathcal{E}_1) &\leq& \frac{\mathbb{E}[S_n^2]}{E_{0n}^2}\notag\\
&=& \frac{4(2\mathbb{E}[E_1]^2 + \sigma_E^2 )}{K_{\varepsilon,\lambda}^2 \mathbb{E}[E_1]^2}.
\label{ehviol2}
\end{IEEEeqnarray}

Unlike (\ref{egout}), the RHS above is independent of $n$. However, by a clever choice of $K_{\varepsilon,\lambda}$, it can be made small enough. Our choice of $K_{\varepsilon,\lambda}$ will ensure that $Pr(\mathcal{E}_1) \leq (1-\lambda )\varepsilon$. Thus a total of $N_n + n$ slots are used for both saving and transmission in this scheme.

We'd like to remark that the aforementioned results do not assume that $V_i$ is Gaussian and the channel part has no role here except for the input constraint. This means that the above bound holds for non-Gaussian energy harvesting channels with independent inputs having variance $\mathbb{E}[E_1]$.

\subsection{Lower bound derivation}
Let $\mathcal{E}_H = \mathcal{E}_0 \cup \mathcal{E}_1$. Now in maximal probability of error (see (\ref{maxpoe})), we see that
\begin{IEEEeqnarray}{rCl}
P_{e,max} &=& \max_{1\leq i\leq M} Pr[\hat{U}\ne i|U=i] \notag\\
&=&  \max_{1\leq i\leq M} Pr[\hat{U}\ne i,\mathcal{E}_H^c|U=i] + Pr[\hat{U}\ne i,\mathcal{E}_H|U=i]  \notag\\
&\leq& \max_{1\leq i\leq M} Pr[\hat{U}\ne i|U=i,\mathcal{E}_H^c] + Pr[\mathcal{E}_H]  
\label{maxpsim}
\end{IEEEeqnarray}

At this point, we invoke Feinstein's lemma (see \cite{PPV1}) which is stated as follows.

\begin{lem}[Feinstein's Lemma]\label{flem}
Let $\varepsilon>0$, $n\geq 1$ and $P_{\mathbf{X}^n}$ be given. Then there exists a $(n,M,\varepsilon)$, maximal p.o.e. code such that for any $\gamma_n >0$

\begin{equation}
\varepsilon \leq Pr\left[\log \left(\frac{W^n(\mathbf{Y}^n|\mathbf{X}^n)}{P_{\mathbf{Y}^n}(\mathbf{Y}^n)}\right)\leq \log \gamma_n\right] +\frac{M}{\gamma_n}, 
\label{fleq}
\end{equation}
where $P_{\mathbf{Y}^n} = \sum\limits_{\mathbf{x}^n} W^n(.|\mathbf{x}^n)P_{\mathbf{X}^n}(\mathbf{x}^n) $.
 
\end{lem}
 
In Feinstein's Lemma above, we pick $P_X$ as Gaussian with zero mean and variance $E[E_1]$. Observe that this choice of input distribution allows bound (\ref{ehviol2}) to be valid. Moreover under $\mathcal{E}_H^c$, i.e. absence of outage, it is as if we are transmitting on a Gaussian channel with noise variance $\sigma^2$ and average power constraint $E[E_1]$. Thus, the first term on RHS of (\ref{maxpsim}) is upper bounded by RHS of Feinstein's Lemma. We have already derived an upper bound on  $Pr(\mathcal{E}_H ) $ via (\ref{egout}), (\ref{ehviol2}) and the union bound. Next we shall derive a suitable upper bound on $Pr\left[\log \left(\frac{W^n(\mathbf{Y}^n|\mathbf{X}^n)}{P_{\mathbf{Y}^n}(\mathbf{Y}^n)}\right)\leq \log \gamma_n\right]$.

Let $G_i = \log \left(\frac{\mathbf{W}(Y_i|X_i)}{\mathbf{P}_Y(Y_i)} \right)$. Then we have
\begin{equation}
Pr\left[\log \left(\frac{W^n(\mathbf{Y}^n|\mathbf{X}^n)}{P_{\mathbf{Y}^n}(\mathbf{Y}^n)}\right)\le \log \gamma_n\right] = Pr\left\{ \sum_{i=1}^n G_i \leq \log \gamma_n \right\}. 
\label{erro4}
\end{equation}

Note that $G_i$ are i.i.d based on the remarks provided earlier. Moreover, we have
\begin{IEEEeqnarray}{rCl}
C_{EG} &:=E[G_i] &= \frac{1}{2}\log\left(1+ \frac{\mathbb{E}[E_1]}{\sigma^2}\right), \\
V_{EG} &:=Var(G_i) &= \frac{\mathbb{E}[E_1]}{\mathbb{E}[E_1] + \sigma^2}\log^2_2(e) \label{VEG1}.
\end{IEEEeqnarray}
Also the third moment, $E[|G_i|^3]$, is finite. To proceed further, we state the Berry Esseen's theorem (see Theorem 6.4.1 in \cite{Athr}). 
\begin{lem}[Berry Esseen's Theorem]
Let $X_i,~1\leq i\leq n$, be an i.i.d. sequence of random variables with mean $\mu$, variance $\sigma^2 <\infty$ and $E[|X_1|^3] < \infty$. Let $S_n = \sum\limits_{i=1}^n X_i$. Then we have, for any $x \in \mathbb{R}$,
\begin{equation*}
\left|Pr\left( \frac{S_n - n\mu}{\sigma\sqrt{n}} \leq x\right) - \Phi(x)\right| \leq C\frac{E|X_1 - \mu|^3}{\sigma^3\sqrt{n}},
\label{berry}
\end{equation*}
where $C<1/2$ (see \cite{Tyu}).  Note that the bound is uniform in $x$. 
\end{lem}

 Let $K = \frac{E[|G_i -E[G_i]|^3]}{2V_{EG}^{3/2}}$. Applying Berry Esseen's theorem, we have for any $u \in \mathbb{R}$,
\begin{equation*}
\left|Pr\left\{ \frac{\left(\sum \limits_{i=1}^n G_i\right) -nC_{EG}}{\sqrt{nV_{EG}}} \leq u \right\} - \Phi(u)\right| \leq \frac{K}{\sqrt{n}}\label{bery1}.
\end{equation*}
Substituting $u =  \frac{\log \gamma_n -nC_{EG}}{\sqrt{nV_{EG}}}$, we get

\begin{equation}
Pr\left\{ \sum_{i=1}^n G_i \leq \log \gamma_n \right\} \leq \Phi \left(\frac{ \log \gamma_n - nC_{EG} }{\sqrt{nV_{EG}}}\right) + \frac{K}{\sqrt{n}} \label{bery2}.
\end{equation}

Let $\varepsilon_n = \lambda\varepsilon - \frac{8\sigma_E^2}{K_{\varepsilon,\lambda}\mathbb{E}[E_1]^2\sqrt{n}} - \frac{2K}{\sqrt{n}}$ and $\log \gamma_n = nC_{EG} + \sqrt{nV_{EG}}\Phi^{-1}(\varepsilon_n)$. We pick $n$ large enough to ensure $\varepsilon_n >0$. From (\ref{maxpsim}), (\ref{fleq}), (\ref{erro4}) and (\ref{bery2}) we have
\begin{IEEEeqnarray}{rCl}
\log M &\geq& \log \gamma_n -\frac{1}{2}\log n +O(1) \notag\\
&\geq& nC_{EG} + \sqrt{nV_{EG}}\Phi^{-1}(\varepsilon_n) -\frac{1}{2}\log n +O(1)\notag
%\label{Mboun}
\end{IEEEeqnarray}

We further simplify $\Phi^{-1}\left(\varepsilon_n\right)$ using Taylor's theorem. There exists $u \in (\varepsilon_n , \lambda\varepsilon)$ such that
\begin{equation*}
f\left(\varepsilon_n\right) = f\left(\lambda\varepsilon\right) +(\varepsilon_n-\lambda\varepsilon)f'(u),
\label{tayl2}
\end{equation*}
where $f(x) = \Phi^{-1}(x)$. Note that $f(x)$ has a derivative that is positive, strictly decreasing upto $x=1/2$; beyond which it increases. Thus in $(\varepsilon_n, \lambda\varepsilon)$, $f'(u) \leq \hat{f} = \max\{f'(\varepsilon_{n_0}), f'(\lambda\varepsilon)\}$ where $n_0$ is the smallest $n$ for which $\varepsilon_n > 0$. Hence we get, with our choice of $\varepsilon_n$, that
\begin{equation*}
\log M \geq nC_{EG} + \sqrt{nV_{EG}}\Phi^{-1}\left(\lambda\varepsilon\right) -\frac{1}{2}\log(n) + O(1).
\label{bou1}
\end{equation*}

Let $\hat{n} = n +N_n$. We have used $\hat{n}$ slots out of which $n$ were for data transmission. We will express the result as a function of $\hat{n}$; the total number of slots used. Hence we have,
\begin{IEEEeqnarray}{rCl}
\log M^*(\hat{n}, \varepsilon) &\geq& (\hat{n}-N_n))C_{EG} + \sqrt{nV_{EG}}\Phi^{-1}\left(\left(\lambda\varepsilon\right)\right) -\frac{1}{2}\log((\hat{n}-N_n)) + O(1), \notag \\
&\geq& \hat{n}C_{EG} -K_{\varepsilon,\lambda}\sqrt{\hat{n}}C_{EG}+ \sqrt{nV_{EG}}\Phi^{-1}\left(\lambda\varepsilon\right) - \frac{1}{2}\log\hat{n} + O(1). \label{lbawg}
\end{IEEEeqnarray}
Note that $\sqrt{n} \leq \sqrt{\hat{n}}$ and $\sqrt{n} \geq \sqrt{\hat{n}}-\frac{K_{\varepsilon,\lambda}}{2}$, the latter follows from
\begin{equation}
\sqrt{\hat{n}} = \sqrt{n+K_{\varepsilon,\lambda} \sqrt{n}}= \sqrt{n}\sqrt{1+\frac{K_{\varepsilon,\lambda}}{\sqrt{n}}}\leq \sqrt{n}\left(1+\frac{K_{\varepsilon,\lambda}}{2\sqrt{n}}\right)= \sqrt{n} + \frac{K_{\varepsilon,\lambda}}{2},
\label{sqrn}
\end{equation}
where we have used $(1+x)^{\frac{1}{2}} \leq 1 + \frac{x}{2}$ for $x>0$. From (\ref{lbawg}) and (\ref{sqrn}), we observe that regardless of the sign of $\Phi^{-1}(\lambda\varepsilon)$, the lower bounds obtained differ by a constant which does not depend on $n$. Putting it all together, we get for $\hat{n}$ large enough
\begin{equation}
\log M^*(\hat{n}, \varepsilon) \geq \hat{n}C_{EG}+ \sqrt{\hat{n}}\left[\sqrt{V_{EG}}\Phi^{-1}\left(\lambda\varepsilon\right)-K_{\varepsilon,\lambda}C_{EG}\right] - \frac{1}{2}\log\hat{n} + O(1).\notag
%\label{ehawgnl1}
\end{equation}
This concludes the proof of the achievable bound for Theorem \ref{TH1}.

\section{Finite blocklength Achievability Bound for EH-DMC}\label{FBDMC}
We use the same random coding strategy as in the EH-AWGN channel case. Choose any input distribution $P_X \in \mathcal{F}_{\mathbb{E}[E_1]}$ . Generate an $M\times n$ matrix with each element distributed i.i.d. with distribution $P_X$. Now follow the proof exactly as in the achievability of the EH-AWGN channel case, replacing the term $X_i^2$ with $\Lambda(X_i)$ wherever it is encountered. 

In particular, we could substitute $P_X^* \in \Gamma$ (where $\Gamma$ is the set of capacity achieving input distributions that are contained in $\mathcal{F}_{\mathbb{E}[E_1]}$)  to obtain the best bound. If there are many capacity achieving distributions, then $V(P_X^*;W)$ may change with the choice of distribution $P_X^*$. Hence consider 
\begin{IEEEeqnarray}{rCl}
V_{ED} =\left\{ \begin{array}{c c}
V_{\min} := \min \limits_{P \in \Gamma} V(P;W), & \mbox{if } \varepsilon \leq \frac{1}{2\lambda},\notag\\
V_{\max} := \max \limits_{P \in \Gamma} V(P;W), & \mbox{if } \varepsilon > \frac{1}{2\lambda}.\notag
\end{array} \right.
%\label{edmcv}
\end{IEEEeqnarray}
Putting it all together, we obtain the following achievability bound, 
\begin{equation}
\log M^*(\hat{n},\varepsilon) \geq \hat{n}C_{ED} - \sqrt{\hat{n}}K_{\varepsilon,\lambda} C_{ED} + \sqrt{\hat{n}V_{ED}}\Phi^{-1}\left(\lambda\varepsilon\right) -\log \hat{n} +O(1).\notag
%\label{achdmc2}
\end{equation}
for all $\hat{n}$ sufficiently large,

\section{Converse Theorems}
In this section, we will provide a general upper bound on finite blocklength rates for energy harvesting channels. We resort to methods used in \cite{PPV1} to derive these new bounds. Then we apply these to the EH-AWGN and the EH-DMC.

We recall the following error probability function $\beta_{\alpha}(P,Q)$ (see \cite{PPV1}).
\begin{defi} 
Given two distributions $P$ and $Q$ on $\mathcal{X}$, define for $\alpha \in [0,1]$,
\begin{equation}
\beta_{\alpha}(P,Q) := \min Q[T=1] := \min \int_\mathcal{X}P_{T|X}(1|x)dQ(x)
\label{betax}
\end{equation}
where the minimum is over all distributions ($P_{T|X}$) of test functions $T:\mathcal{X} \to\{0,1\}$ such that $P[T=1]\geq \alpha$.
\end{defi}
This function is essentially the type 2 error probability (probability of deciding $P$ when $Q$ is true) when the type 1 error probability is less than $1-\alpha$. 

The Meta Converse, proved in \cite{PPV1}, is one of the tightest known general converse bounds for any channel. There are two versions, one for average probability of error and the other for maximal probability of error. Note that these are single shot bounds and can be naturally extended for blocklength $n$.
\begin{lem}[Meta Converse (avg p.o.e)]\label{metco2}
Every $(M,\varepsilon)$ average probability of error code satisfies 
\begin{equation}
 M \leq \sup_{P_X} \frac{1}{\beta_{1-\varepsilon}(P_{XY},P_X Q_Y)}.\notag
%\label{meta22}
\end{equation}
for any output distribution $Q_Y$.
\end{lem}

\begin{lem}[Meta Converse (max p.o.e)]\label{metco}
Every $(M,\varepsilon)$ maximal probability of error code satisfies 
\begin{equation}
 M \leq \frac{1}{\beta_{1-\varepsilon}(P_{Y|X=c(\overline{m})},Q_Y)}\leq \sup_{x \in \mathbb{F}}\frac{1}{\beta_{1-\varepsilon}(P_{Y|X=x},Q_Y)}.\notag
%\label{meta2}
\end{equation}
for any output distribution $Q_Y$ and codewords coming from $\mathbb{F} \subset \mathcal{X}$, where $\mathcal{X}$ is the input alphabet and $c(\overline{m})$ is the codeword of the message $\overline{m}$ satisfying
\begin{equation}
\overline{m} = \arg \min \limits_{m\in [M]} Pr[\hat{U} = m|U=m] .\label{minach}
\end{equation}
under channel $Q_Y$.
\end{lem}

However it is not immediately clear as to the technique of incorporating the effects of energy harvesting in the above expression. This is due to the fact that the set $\mathbb{F}$ above, which is the constrained set, changes with energy. Also unlike traditional channels, the codebook will change depending on available energy. Hence any codeword is of the form $c(m,\mathbf{e})$ for message $m$ and energy vector $\mathbf{e}$. 

\subsection{Energy Harvesting Converse (General Version)}
Under the energy harvesting setup described earlier, we obtain the following converse bounds. 
\begin{theo}
Given an energy harvesting setup with channel $W$, incoming energy process $E\sim P_E$ i.i.d., every $(M,\varepsilon)$ code (average p.o.e) satisfies
\begin{equation}
M \leq \sup_{P_{X^n|E^n}} \frac{1}{\beta_{1-\varepsilon}(P_{E^nX^nY^n},P_{E^nX^n}Q_{Y^n})}
\label{metaeh1}
\end{equation}
 where $P_{E^nX^nY^n}(e^n,x^n,y^n) = P_{E^n}(e^n)P_{X^n|E^n}(x^n|e^n)W(y^n|x^n)$ and for any output distribution $Q_{Y^n}$. The supremum is taken over all distributions that satisfy the energy harvesting constraints. Under the maximal probability of error case, we have
\begin{equation}
M \leq \frac{1}{\beta_{1-\varepsilon} \left(W(.|c(\overline{m},*))P_{E^n}(*),Q_{Y^n}P_{E^n}\right)} 
%\leq \frac{1}{\beta_{1-\varepsilon}\left(\mathbb{E}\left[\sup\limits_{x \in \mathbb{F}_E}W(.|x)\right],Q_Y\right) }
\label{metaeh2}
\end{equation}
for any output distribution $Q_{Y^n}$ and $c(\overline{m},*)$ is the codeword whose message $\overline{m}$ satisfies (\ref{minach}). Here $.$ represents the output alphabet and $*$ represents the energy alphabet.
%Here $\mathbb{F}_E \subseteq \mathcal{X}$ depicts the energy harvesting constraints and hence depends on the energy process $E$. Note that the RHS is valid only if $\mathbb{E}[\sup\limits_{x \in \mathbb{F}_E}W(.|x)]$ is a measure on output alphabet $\mathcal{Y}$ which is true for EH-AWGN and EH-DMC.
\end{theo}
\begin{proof}
The proof of (\ref{metaeh1}) is available in \cite{TanF2}. For the proof of (\ref{metaeh2}), refer to Appendix \ref{stpr}.
\end{proof}
The bound in (\ref{metaeh1}) was used to develop a finite blocklength converse for EH-AWGN channels, extended to the block i.i.d. energy arrivals regime \cite{TanF2}. We shall derive the same result for EH-AWGN channels under maximal probability of error criterion but using (\ref{metaeh2}). 

There is a weaker, but analytically convenient, converse bound under maximal p.o.e. stated as follows.
\begin{theo}\label{Thuse}
Consider an energy harvesting setup with channel $W$, incoming energy process $E\sim P_E$ i.i.d. and cost function $\Lambda$ as defined in Section \ref{subdm}. Under the requirement that every codeword $\mathbf{x}(m,\mathbf{e}^n)$ satisfying the energy harvesting constraint, i.e., 
\begin{equation}
\sum_{i=1}^n \Lambda(x_i(m,\mathbf{e}^n))  \leq \sum_{i=1}^n e_i,
\label{ehconst3}
\end{equation}
for energy vector $\mathbf{e}^n$ and maximal probability of error $\varepsilon$,
\begin{equation}
M \leq \sup_{\mathbf{x}^n \in \mathbb{F}_{\overline{E}_n}}\frac{1}{\beta_{1-\varepsilon-\tau_n} \left(W(.|\mathbf{x}^n),Q_{Y^n}\right)},
\label{metaeh3}
\end{equation}
where $\tau_n = Pr(\sum_{i=1}^n E_i \geq n\overline{E}_n)$,
\begin{equation}
\mathbb{F}_{\overline{E}_n} = \left\{\mathbf{x}^n: \sum_{i=1}^n \Lambda(x_i)  \leq n\overline{E}_n\right\},
\label{alset}
\end{equation}
and $\overline{E}_n$ is a non-negative sequence chosen such that $\tau_n<1-\varepsilon $.
\end{theo}
\begin{proof}
Refer Appendix \ref{Thusea}.
\end{proof}
There is a nice structure for EH-AWGN channels that helps in getting sharper bounds when using (\ref{metaeh1}) or (\ref{metaeh2}). These details are clarified in the proof of the converse bound for EH-AWGN channel. However that structure is absent when dealing with EH-DMCs. Theorem \ref{Thuse}, will be used to get a useful upper bound in this case.

\section{Finite blocklength converse bound for EH-AWGN}\label{awgconv}
We argue that it suffices to look at codewords $\mathbf{x}^n$ that satisfy 
\begin{equation}
\sum_{k=1}^n x_k^2 = \sum_{k=1}^n e_k,
\label{consteq}
\end{equation}
where $\mathbf{e}^n$ is the energy vector. In short, we are ignoring the outage events that can happen for $1\leq k<n$ and we are using up all the energy in transmission at time $n$. The former is justified by noting that doing so merely relaxes the constraints and that can only increase capacity. Hence any upper bound on the relaxed version is an upper bound on the original version. As for the latter, it is a well known Yaglom-map trick where given the best code of codeword length $n$ but satisfying (\ref{consteq}) with a strict inequality ($<$), we can construct a new code with the same probability of error but with codeword length $n+1$. The extra symbol is picked so as to exhaust all remaining energy. This new code clearly satisfies (\ref{consteq}), is an upper bound for the original $n$ length code and is further upper bounded by the largest code of codeword length $n+1$ satisfying (\ref{consteq}).

Let $0<\varepsilon<1$, the maximal probability of error be fixed. Picking $W$ as a Gaussian channel with variance $\sigma^2$ and $Q_{Y^n} = \prod_{i=1}^n Q_Y$, where $Q_Y$ is Gaussian with mean $0$ and variance $\mathbb{E}[E_1]+\sigma^2$. Now for two distributions $P_1$ and $P_2$, and any $\gamma >0$, $\beta_\alpha(P_1,P_2)$ is lower bounded as (from \cite{PPV1}, ($106$))
\begin{equation}
\beta_\alpha(P_1,P_2) \geq  \frac{1}{\gamma} \left(\alpha - P_1\left[\frac{dP_1}{dP_2} \geq \gamma\right]\right).
\label{betlb}
\end{equation}
From (\ref{metaeh2}) and (\ref{betlb}), we have for any $\gamma_n > 0$,
\begin{equation}
M \leq \frac{\gamma_n}{1-\varepsilon - Pr\left[\log \frac{W(\mathbf{Y}^n|\mathbf{x}^n(\overline{m},\mathbf{E}))}{Q_{Y^n}} \geq \log \gamma_n\right]}
\label{Mb11}
\end{equation}
where the probability is under $W(.|\mathbf{x}(\overline{m},*))P_{E^n}(*)$. Since $W$ here is a Gaussian channel, we can replace $Y_i$ with $x_i(\overline{m},\mathbf{e}) + Z_i$ where $Z_i$ are i.i.d. $\mathcal{N}(0,\sigma^2)$. The probability term in the denominator then simplifies to
\begin{IEEEeqnarray}{rCl}
&&Pr\left[\log \frac{W(\mathbf{Y}^n|\mathbf{x}^n(\overline{m},\mathbf{E}))}{Q_{Y^n}} \geq \log \gamma_n\right] \notag\\
&=& Pr\left[\sum_{i=1}^n \frac{(x_i(\overline{m},\mathbf{E})+Z_i)^2}{2(\mathbb{E}[E_1]+\sigma^2)}\log_2(e)  - \sum_{i=1}^n \frac{Z_i^2}{2\sigma^2}\log_2(e) \geq \log(\gamma_n) - nC_{EG}\right] \notag\\
&=&Pr\left[\sum_{i=1}^n \left(\frac{Z_i}{\sigma} - \frac{x_i(\overline{m},\mathbf{E})\sigma}{\mathbb{E}[E_1]}\right)^2 \leq \frac{2(\mathbb{E}[E_1] +\sigma^2)}{\mathbb{E}[E_1]}(nC_{EG}-\log \gamma_n)\ln 2 + \sum_{i=1}^n x_i^2(\overline{m},\mathbf{E}) \left( \frac{\sigma^2}{\mathbb{E}[E_1]^2} + \frac{1}{\mathbb{E}[E_1]}\right)\right] \notag\\
&=&Pr\left[\sum_{i=1}^n \left(\frac{Z_i}{\sigma} - \frac{x_i(\overline{m},\mathbf{E})\sigma}{\mathbb{E}[E_1]}\right)^2 \leq \frac{2(\mathbb{E}[E_1] +\sigma^2)}{\mathbb{E}[E_1]}(nC_{EG}-\log \gamma_n)\ln 2 + \sum_{i=1}^n E_i \left( \frac{\sigma^2}{\mathbb{E}[E_1]^2} + \frac{1}{\mathbb{E}[E_1]}\right)\right] \label{esw}
\end{IEEEeqnarray}
where (\ref{esw}) follows from (\ref{consteq}). Now, we condition the above probability term on $\mathbf{E}=\mathbf{e}$, noting that $\mathbf{E}$ is independent of $\mathbf{Z}$. We observe then that the probability is the cumulative distribution function (CDF) of a non-central $\chi^2$ distribution with $n$ degrees of freedom and non-centrality parameter 
\begin{equation}
B = \sum_{i=1}^n \frac{x_i^2(\overline{m},\mathbf{e})\sigma^2}{\mathbb{E}[E_1]^2} =\sum_{i=1}^n \frac{e_i\sigma^2}{\mathbb{E}[E_1]}.
\label{nccp}
\end{equation}
The CDF of a noncentral $\chi^2$ random variable $\hat{Z}$ equals 
\begin{equation}
Pr(\hat{Z} \leq u) = 1-Q^M_{n/2}(\sqrt{B},\sqrt{u}),
\label{Marc}
\end{equation}
where $Q^M_d(a,b)$ is the Marcum Q function of order $d$ (see \cite{marc}). Now we observe that the CDF does not depend on the individual $x_i$ or $e_i$ but rather on the sum of $e_i$. Replacing $x_i(\overline{m},\mathbf{E})$ with $\sqrt{E_i}$ in (\ref{esw}) will not change the CDF. Hence from \eqref{nccp} and \eqref{Marc}, (\ref{esw}) equals
\begin{equation}
Pr\left[\sum_{i=1}^n \left(\frac{Z_i}{\sigma} - \frac{\sqrt{E_i}\sigma}{\mathbb{E}[E_1]}\right)^2 \leq \frac{2(\mathbb{E}[E_1] +\sigma^2)}{\mathbb{E}[E_1]}(nC_{EG}-\log \gamma_n)\ln 2 + \sum_{i=1}^n E_i \left( \frac{\sigma^2}{\mathbb{E}[E_1]^2} + \frac{1}{\mathbb{E}[E_1]}\right)\right]
\label{esw2}
\end{equation}
This is precisely the structure we mentioned earlier that allows us to work with a simplified expression. As a result, the terms in the summation are i.i.d. (as opposed to just being independent). By suitably rearranging the terms, \eqref{esw2} equals
\begin{equation}
Pr\left[ \frac{\sum_{i=1}^n\eta_i}{\sqrt{nV_{EG2}}} \leq \frac{nC_{EG}-\log \gamma_n}{\sqrt{nV_{EG2}}} \right]
\label{berrys1}
\end{equation}
where $\eta_i$ are i.i.d. with zero mean and variance $V_{EG2} = \frac{\mathbb{E}[E_1]^2 + \mathbb{E}[E_1^2] + 4\sigma^2\mathbb{E}[E_1]}{4(\mathbb{E}[E_1]+\sigma^2)^2}\log_2^2(e)$. The third moment of $\eta_i$ is finite. Applying the Berry Esseen theorem (Lemma \ref{berry}) and picking $\log \gamma_n = nC_{EG} -\sqrt{nV_{EG2}}\Phi^{-1}(\alpha_n)$, where $\alpha_n$ is picked such that $0<\alpha_n<1-\varepsilon$ gives us
\begin{equation}
Pr\left[ \frac{\sum_{i=1}^n\eta_i}{\sqrt{nV_{EG2}}} \leq \frac{nC_{EG}-\log \gamma_n}{\sqrt{nV_{EG2}}} \right] \leq \alpha_n + \frac{\kappa}{\sqrt{n}}
\label{berrys2}
\end{equation} 
where $\kappa = \mathbb{E}[|\eta_i|^3]/V_{EG2}^{3/2}$. 

Pick $\alpha_n = 1-\varepsilon - \frac{2\kappa}{\sqrt{n}}$. For $n$ sufficiently large, $0<\alpha_n<1-\varepsilon$. From (\ref{Mb11}), (\ref{esw2}) and (\ref{berrys2}), we get
\begin{equation}
\log M \leq nC_{EG} -\sqrt{nV_{EG2}}\Phi^{-1}(\alpha_n) - \log(\kappa/\sqrt{n})\notag
%\label{conf11}
\end{equation}
Using Taylor series expansion on $\Phi^{-1}$ as well as bounding steps similar to the proof of achievability of Theorem \ref{TH1}, we obtain
\begin{equation}
\log M \leq nC_{EG} +\sqrt{nV_{EG2}}\Phi^{-1}(\varepsilon) + \frac{1}{2}\log n +O(1) \notag
%\label{conf12}
\end{equation}

\section{Finite blocklength converse bound for EH-DMC}\label{DMConv}
Unfortunately we cannot simply mirror the proof of the EH-AWGN channel converse in Section \ref{awgconv} as the AWGN channel structure that was exploited there is absent here. However, there is a different structure that can be exploited here, namely the method of types  (see \cite{Csiz}). We will be using the framework of Theorem \ref{Thuse}. Let $0<\varepsilon<1$ be given and the DMC of the EH-DMC be denoted by $W(y|x)$. The incoming energy random variables $E_i$ are i.i.d. as before. 
%It may happen that with the energy harvesting system in place, because of the way $\Lambda$ is defined, no input distribution is admissible. If this happens, the capacity is zero and so is the upper bound. 

Recall the definitions given in \eqref{conset0} and \eqref{conset}. We have from (\ref{metaeh3}),
\begin{equation}
M \leq \sup_{\mathbf{x}^n \in \mathbb{F}_{\overline{E}_n}}\frac{1}{\beta_{1-\varepsilon-\tau_n} \left(W(.|\mathbf{x}^n),Q_{Y^n}\right)}.
\label{dmet}
\end{equation}

We pick $\overline{E}_n = \mathbb{E}[E_1]+\delta_n.$ where $\delta_n > 0$. Then $\tau_n$ is given by
\begin{equation}
\tau_n = Pr\left(\sum_{i=1}^n E_i \geq n(\mathbb{E}[E_1]+\delta_n)\right).\notag
%\label{taun}
\end{equation}
We will ensure $\tau_n \leq \frac{\varepsilon}{4}$. To do this, pick $\delta_n = \frac{D_\varepsilon}{\sqrt{n}}$ where $D_{\epsilon} = \sqrt{\frac{4\sigma_E^2}{\varepsilon}}$ and use Chebyshev's inequality.

We can rewrite (\ref{dmet}) as
\begin{equation}
M \leq \sup_{P\in \mathcal{F}_{\overline{E}_n}\cap \mathcal{P}_n} \sup_{\mathbf{x}^n \in T_P}\frac{1}{\beta_{1-\varepsilon-\tau_n} \left(W(.|\mathbf{x}^n),Q_{Y^n}\right)}.
\label{dmet2}
\end{equation}
where $T_P$ denotes the type class of distribution $P$ and $\mathcal{P}_n$ is the set of all types for sequences of length $n$. Consider the inner supremum term,
\begin{equation}
\sup_{\mathbf{x}^n \in \mathcal{T}_P} \frac{1}{\beta_{1-\varepsilon-\tau_n}(W(.|\mathbf{x}^n),Q_{Y^n})}.\notag
%\label{insup}
\end{equation}
The beta error function above is independent of which sequence $\mathbf{x} $ is picked provided that the sequences have the same type \cite{PPV1} and $Q_{Y^n} = \prod_{k=1}^n Q_Y$ for some distribution $Q_Y$ on $\mathcal{Y}$ . Hence pick any sequence $\mathbf{x}$ from $\mathcal{T}_{P_0}$ where $P_0 \in \mathcal{F}_{\overline{E}_n}\cap \mathcal{P}_n$. 
%Note that this argument was used in proving Theorem 48 in \cite{PPV1} while applying the meta converse to a standard DMC. We show that the arguments in \cite{PPV1} are applicable here with only a few minor changes. 

Let $Q_Y=P_0W$. We recall  \cite[Theorem 48]{PPV1} for standard, non-exotic DMCs. Although this bounded the maximal subcode of type $P_0$ of the maximal code, we note that the term actually being bounded is the beta error function as mentioned below.
\begin{lem}\label{polthm}
For $0<\varepsilon < 1$, for all $P_0 \in \mathcal{P}_n$, $\mathbf{x}\in \mathcal{T}_{P_0}$ and $n$ sufficiently large, we have
\begin{equation}
- \log \beta_{1-\varepsilon}(W^n(.|\mathbf{x}),(P_0W)^n) \leq nC_D+\sqrt{nV_D}\Phi^{-1}(\varepsilon) +\frac{1}{2}\log n + O(1)\notag
%\label{polcon1}
\end{equation}
where 
\begin{equation}
V_D=\left\{ \begin{array}{c c}
	V_{\min} = \min \limits_{P\in \Gamma} V(P;W), &  0 <\varepsilon\leq 1/2, \\
	V_{\max}= \max \limits_{P\in \Gamma} V(P;W), & 1/2 < \varepsilon < 1,
\end{array}\right.\notag
%\label{dispdmc}
\end{equation}
and $\Gamma$ is the set of capacity achieving distributions.
\end{lem}

Note that the bound on RHS does not depend on the distribution of the type. Hence if we make the following substitutions:
\begin{enumerate}
	\item Replace $\Gamma$ with
	\begin{equation}
	\Gamma_{\overline{E}_n} = \{P\in \mathcal{F}_{\overline{E}_n}: I(P;W)=C_{ED}\}
	\label{gamcos}
	\end{equation}
This is because the outer supremum in (\ref{dmet2}) is over $\mathcal{F}_{\overline{E}_n}$. Note that the original proof of Lemma \ref{polthm} used the fact that $\Gamma$ was compact and convex. These properties hold for $\Gamma_{\overline{E}_n}$ so we may substitute this wherever $\Gamma$ was used.
\item The final supremum that gives the uniform (over input distributions) bound was over $\mathcal{P}$. Here we substitute $\mathcal{F}_{\overline{E}_n}$ in its place.
\item $\varepsilon$ is replaced by $\varepsilon + \tau_n$.
\end{enumerate}
then
\begin{equation}
\log M^*(n,\varepsilon) \leq nC_D(\overline{E}_n) +\sqrt{n\hat{V}(\overline{E}_n)}\Phi^{-1}\left(\varepsilon + \tau_n\right) +O(\log(n)),
\label{DMConf}
\end{equation}
where $C_D(.)$ is defined in \eqref{cost1} and
\begin{equation}
\hat{V}(\overline{E}_n) = \left\{ \begin{array}{c c}
	V_{\min}^{(n)} = \min \limits_{P\in \Gamma_{\overline{E}_n}} V(P;W), &  0 <\varepsilon + \tau_n \leq 1/2, \\
	V_{\max}^{(n)}= \max \limits_{P\in \Gamma_{\overline{E}_n}} V(P;W), & 1/2 < \varepsilon+ \tau_n < 1.
\end{array}\right.
\label{disdef}
\end{equation}
We can further simplify (\ref{DMConf}) by expanding $C_D(\overline{E}_n)$, $\hat{V}(\overline{E}_n)$ and $\Phi^{-1}(u)$. 

Now $C_D(a)$ is a non-decreasing concave function (see \cite{Csiz}). Hence we have for any $a>0, b > 0$,
\begin{equation}
C_D(a+b) \leq C_D(a) + bC_D'(a),\notag
%\label{concav}
\end{equation}
where $C_D'(.)$ is the derivative of $C_D(a)$. Let $a=\mathbb{E}[E_1]$ and $b=\delta_n$. Note that $C_D'(a)$ in this case is a constant since $\mathbb{E}[E_1]$ is a constant. 

Using Taylor series expansion, we get that for some constant $K_\varepsilon$,
\begin{equation}
\Phi^{-1}(\varepsilon + \tau_n) \leq \Phi^{-1}(\varepsilon) + \tau_nK_\varepsilon.\notag
%\label{phit}
\end{equation}
Now let $\varepsilon_R$ be the root of 
\begin{equation}
\Phi^{-1}(\varepsilon) + \frac{K_{\varepsilon}\varepsilon}{4}=0.\notag
%\label{rooty}
\end{equation}
Pick any $\eta>0$. Observe that for $n$ sufficiently large, $\Gamma_{\overline{E}_n} \subset \Gamma_{\mathbb{E}[E_1]+\eta}$. Hence we can replace $\hat{V}(\overline{E}_n)$ with
\begin{equation}
V^*_{\varepsilon}(\eta) = \left\{ \begin{array}{c c}
	  \min \limits_{P\in \Gamma_{\mathbb{E}[E_1]+\eta}} V(P;W), &  0 <\varepsilon \leq \varepsilon_R, \\
	 \max \limits_{P\in \Gamma_{\mathbb{E}[E_1]+\eta}} V(P;W), &  \varepsilon_R< \varepsilon < 1.
\end{array}\right.\notag
%\label{Vfinale}
\end{equation}
Note that $C_D(\mathbb{E}[E_1]) \equiv C_{ED}$, Thus we have for $n$ sufficiently large
\begin{equation}
\log M^*(n,\varepsilon) \leq nC_{ED} +\sqrt{n}C'(\mathbb{E}[E_1])D_{\varepsilon} +   \sqrt{nV^*_{\varepsilon}(\eta)}\left(\Phi^{-1}(\varepsilon) + \frac{K_{\varepsilon}\varepsilon}{4} \right) +O(\log n).\notag
%\label{dmcfinq}
\end{equation}

\section{Moderate Deviation Asymptotics}
In this section, we discuss the bounds on the moderate deviation asymptotics for the EH-AWGN channel and the EH-DMC. In this analysis, unlike in the second order analysis in the previous sections, we allow probability of error to go to zero as a function of blocklength $n$. However we do so in the moderate deviations regime which is defined formally as follows (see \cite{polymd}).

\begin{defi}[Moderate Deviation coefficient]
Given a channel $W$, let $\rho_n$ be a sequence of non-negative real numbers such that $\rho_n \to 0$ and $n\rho_n^2 \to \infty$. Then for codes of size $M_n$ satisfying $\log M_n = n(C - \rho_n)$, where $C$ is the channel capacity, the moderate deviations coefficient (MDC) $\xi$, if it exists, is defined as 
\begin{equation}
\xi = \lim_{n \to \infty} \frac{\log \varepsilon(n)}{n\rho_n^2}, \notag
%\label{mdev}
\end{equation}
where $\varepsilon(n)$ is the probability of error as a function of blocklength $n$.
\end{defi}

For memoryless channels with channel dispersion $V>0$, it was shown in \cite{polymd} that $\xi = -\frac{1}{2V}$ is the moderate deviation coefficient. In case of energy harvesting channels, it is more involved. This is due to not knowing the exact dispersion value as well as the fact that energy harvesting channels are not truly memoryless due to the energy vector. However they have a part which is memoryless and this is what we have been exploiting so far in our analysis. 

\subsection{MDC for EH-AWGN channels}
We now state the following theorem bounding the MDC for EH-AWGN channels.
\begin{theo}
For an EH-AWGN channel with energy process $E_i$ i.i.d. with variance $\sigma_E^2$, the MDC is bounded as
\begin{IEEEeqnarray}{rCl}
\liminf_{n\to \infty} \frac{\log \varepsilon(n)}{n\rho_n^2} &\geq& -\frac{1}{2V_{EG2}}, \label{awmdlb}\\
\limsup_{n\to \infty} \frac{\log \varepsilon(n)}{n\rho_n^2} &\leq& -\frac{1}{2V_{EG}}, \label{awmdub}
\end{IEEEeqnarray}
where $V_{EG}$ is defined in (\ref{ehawgnclb}) and $V_{EG2}$ is defined in (\ref{ehawgncub}).
\end{theo}
\begin{proof}
To show (\ref{awmdlb}), let us consider (\ref{Mb11}) whose terms are rearranged, replacing $\varepsilon$ with $\varepsilon(n)$, as
\begin{equation}
\varepsilon(n) \geq  Pr\left[\log \frac{W(\mathbf{Y}^n|\mathbf{x}^n(\overline{m},\mathbf{E}))}{Q_{Y^n}} \leq \log \gamma_n\right] - \frac{\gamma_n}{M}.\notag
%\label{mden1}
\end{equation} 

We also have from (\ref{berrys1}) that 
\begin{equation}
 Pr\left[\log \frac{W(\mathbf{Y}^n|\mathbf{x}^n(\overline{m},\mathbf{E}))}{Q_{Y^n}} \leq \log \gamma_n\right] = Pr\left[ \sum_{i=1}^n\eta_i \geq nC_{EG}-\log \gamma_n \right].\notag
%\label{mden2}
\end{equation}
Now let $\log M = n(C_{EG}-\rho_n)$ and $\log \gamma_n = n(C_{EG}-\alpha\rho_n)$ for any $\alpha>1$. From \cite[Theorem 3.7.1]{demb}, we get
\begin{equation}
\liminf_{n\to \infty} \frac{ \log Pr\left[ \sum_{i=1}^n\eta_i \geq nC_{EG}-\log \gamma_n \right]}{n\rho_n^2} \geq -\inf_{x\geq\alpha} \frac{x^2}{2V_{EG2}} = -\frac{\alpha^2}{2V_{EG2}},  \notag
%\label{mden3}
\end{equation}
where noting that $V_{EG2}$ is the variance of $\eta_i$ and letting $\alpha \to 1$, we get (\ref{awmdlb}).

To show (\ref{awmdub}), we need to modify some of our arguments which we used while discussing the save and transmit scheme. This is because we need to show that codes of $\log M = n(C_{EG} - \rho_n)$ exist. The analysis so far was done so as to work with the optimum order of $\sqrt{n}$. This is not valid anymore as $\rho_n > 1/\sqrt{n}$.

%\begin{enumerate}
	%\item As $\varepsilon(n)$ is now a function of $n$, several arguments that assumed it to be constant are no longer valid. For example, we had defined $N_n = K_{\varepsilon,\lambda}\sqrt{n}$. But now as $K_{\varepsilon(n),\lambda}$ is also a function of $n$, if we are not careful, $N_n$ may grow as large or larger than $n$ which is not feasible.
%\end{enumerate}

Recalling error events $\mathcal{E}_0$ from (\ref{egout}) and $\mathcal{E}_1$ from (\ref{ehviol}), we will show that with an appropriate choice for $N_n$ and $E_{0n}$, we can set
\begin{equation}
Pr(\mathcal{E}_0)+ Pr(\mathcal{E}_1) \leq \frac{\varepsilon(n)}{2}.\notag
%\label{errmd1}
\end{equation}
To ensure this, let us choose
\begin{equation}
N_n = \max\left\{ \frac{16\sigma_E^2}{\varepsilon(n)\mathbb{E}[E_1]^2}, \frac{4\sqrt{n(2\mathbb{E}[E_1]^2+\sigma_E^2)}}{\mathbb{E}[E_1]\sqrt{\varepsilon(n)}}\right\}.\notag  
%\label{Nnn}
\end{equation}
Clearly $N_n \to \infty$ as $n \to \infty$ and both $Pr(\mathcal{E}_0)$ and $Pr(\mathcal{E}_1)$ are each upper bounded by $\varepsilon(n)/4$. 

Hence the probability of error $\varepsilon(n)$ is bounded by
\begin{IEEEeqnarray}{rCl}
\varepsilon(n) &\leq& \frac{\varepsilon(n)}{2} +  Pr\left[\log \left(\frac{W^n(\mathbf{Y}^n|\mathbf{X}^n)}{P_{\mathbf{Y}^n}(\mathbf{Y}^n)}\right)\leq \log \gamma_n\right] +\frac{M}{\gamma_n},\notag\\
\frac{\log (\varepsilon(n)/2)}{n\rho_n^2} &\leq&\frac{1}{n\rho_n^2} \log \left[Pr\left\{ \sum_{i=1}^n G_i \leq \log \gamma_n \right\} + 2^{-(1-\alpha)n\rho_n}\right]\notag
%\label{en2}
\end{IEEEeqnarray}

Now let $\log \gamma_n = n(C_{EG} - \alpha\rho_n)$ where $\alpha<1$ and $\log M = n(C_{EG}-\rho_n)$. Codes of the latter size are assured by Feinstein's lemma. Now from (\ref{erro4}) and  \cite[Theorem 3.7.1]{demb}, we have

\begin{equation}
\limsup_{n \to \infty} \frac{1}{n\rho_n^2} \log Pr\left\{ \sum_{i=1}^n G_i \leq \log \gamma_n \right\} \leq -\inf_{x\leq-\alpha} \frac{x^2}{2V_{EG}} = -\frac{\alpha^2}{2V_{EG}}.
\label{en3}
\end{equation}
Letting $\alpha\to 1$, we get (\ref{awmdub}).
\end{proof}

%% Put proof of DMC here.
\subsection{MDC for EH-DMC}
The MDC bounds for EH-DMC should be analogous to that of the EH-AWGN channel. However since $V_{ED}$ varies with the choice of $\lambda$, we need to refine it slightly. 
\begin{theo}
For the EH-DMC, the following bounds on MDC apply.
\begin{IEEEeqnarray}{rCl}
\liminf_{n\to \infty} \frac{\log \varepsilon(n)}{n\rho_n^2} &\geq& -\inf_{\eta>0}\frac{1}{2V_{\min,\eta}}, \label{dmdlb}\\
\limsup_{n\to \infty} \frac{\log \varepsilon(n)}{n\rho_n^2} &\leq& -\frac{1}{2V_{\min}}, \label{dmdub}
\end{IEEEeqnarray}
where $V_{\min} = \min \limits_{P \in \Gamma_{\mathbb{E}[E_1]}} V(P;W)$ and $V_{\min,\eta} = \min \limits_{P \in \Gamma_{\mathbb{E}[E_1]+\eta}} V(P;W)$ where $\Gamma$ is the set of capacity achieving input distributions that are in $\mathcal{F}_{\mathbb{E}[E_1]}$.
\end{theo}
\begin{proof}
Bound (\ref{dmdlb}), follows from \cite[Theorem 6]{polymd} with the following changes:

\begin{enumerate}
	\item The distributions need to be admissible i.e. from $\mathcal{F}_{\overline{E}_n}$.
	\item $\varepsilon(n)$ is to be replaced with $\varepsilon(n)+\tau_n$. But as per our construction, $\tau_n < \varepsilon(n)/4$. Hence it is same as replacing $\varepsilon(n)$ with $\frac{5}{4}\varepsilon(n)$.
\end{enumerate}

%consider any capacity achieving admissible distribution $P^*$. Then because $\mathcal{F}_{\overline{E}_n}$ is a compact set, there is a sequence of admissible types $P_{0n}$ that converge to $P^*$. As $V(P;W)$ is continuous, this will imply $V(P_{0n};W) \to V(P^*;W)$.
%
%Recall (\ref{insup}) for an admissible type $P_{0n}\in \mathcal{P}_n$. Since the expression is independent of the sequence in the type we have for $\mathbf{x}\in T_{P_{0n}}$ and (\ref{betlb}) that
%\begin{equation}
%\varepsilon_n \geq 1-\tau_n -Pr\left[\log \frac{W(\mathbf{Y}|\mathbf{x})}{P_{0n}W(\mathbf{Y})}\geq \log \gamma_n\right] - \frac{\gamma_n}{M}.
%\label{enc1}
%\end{equation}
%Note that we had set $\tau_n < \varepsilon/4$. Substituting this bound here gives us
%\begin{equation}
%\frac{5}{4}\varepsilon_n \geq Pr\left[\log \frac{W(\mathbf{Y}|\mathbf{x})}{P_{0n}W(\mathbf{Y})}\leq \log \gamma_n\right] - \frac{\gamma_n}{M}.
%\label{enc2}
%\end{equation}
%Let $\log M = n(C_{ED} - \rho_n)$ and for any $\alpha>1$, $\log \gamma_n = n(C_{ED} - \alpha\rho_n)$. Now in \cite[Theorem 6]{polymd}, it was shown that 
%\begin{equation}
%\liminf_{n \to \infty}\frac{1}{n\rho_n^2}\log\left( Pr\left[\log \frac{W(\mathbf{Y}|\mathbf{x})}{P_{0n}W(\mathbf{Y})}\leq \log \gamma_n\right] - \frac{\gamma_n}{M}\right) \geq - \frac{\alpha^2}{2V(P^*;W)}
%\label{enc3}
%\end{equation}

To prove (\ref{dmdub}), we note that the steps are very similar to the proof of (\ref{awmdub}). To begin with, pick a capacity achieving distribution $P_X$ and follow all the steps exactly as before. We get 
\begin{equation}
\limsup_{n\to \infty} \frac{\log \varepsilon(n)}{n\rho_n^2} \leq -\frac{1}{2V(P_X;W)}. \notag%\label{dmdub2}
\end{equation}
Since this is valid for any $P_X \in \Gamma$, the tightest bound is obtained when we replace $V(P_X;W)$ with $V_{\min}$.

\end{proof}

\section{Numerical Results}
We now evaluate and plot the finite blocklength bounds on rate as well as the slots consumed in the saving part of save and transmit as a function of blocklength. We use the formulae derived in the earlier sections, for a specified set of parameters, to evaluate the aforementioned quantities. For the EH-DMC, we describe an energy harvesting binary symmetric channel (BSC) and a binary erasure channel (BEC) and plot the corresponding bounds for these. Note that in all plots, we are ignoring the constant terms in the bounds i.e. coefficients of $O(1/n)$ in the rates. Additionally, we compare our results with the finite blocklength lower bounds of an equivalent non-energy harvesting channel. For example, in the EH-AWGN case, we consider an AWGN channel with average power constraint $\mathbb{E}[E_1]$, while in the EH-DMC cases, we consider corresponding DMCs with power constraint $\mathbb{E}[E_1]$. This will allow us, when the equivalent channel's lower bound is above the energy harvesting upper bound, to comment on the effects of energy harvesting on rate. The gap in rates mentioned henceforth will be the difference between the bounds divided by the upper bound, expressed as a percentage.

\subsection{EH-AWGN results}
We take the maximal probability of error, $\varepsilon = 0.1$, $\mathbb{E}[E_1] = 1$ and $\sigma_E^2 = 5$. We consider blocklengths $n$ between $5000$ to $10000$. We consider three different regimes i.e.,
\begin{enumerate}
	\item Low SNR ($-20$ dB). In this regime (Fig. \ref{Plotb1}), we observe that the lower bound is a poor approximation to the finite blocklength rate. Due to a larger number of errors, this regime also requires more slots to harvest energy to lower the error due to outage (about $20.5 \%$ to $27.6\%$).
	\begin{figure}
		\includegraphics[width=\columnwidth]{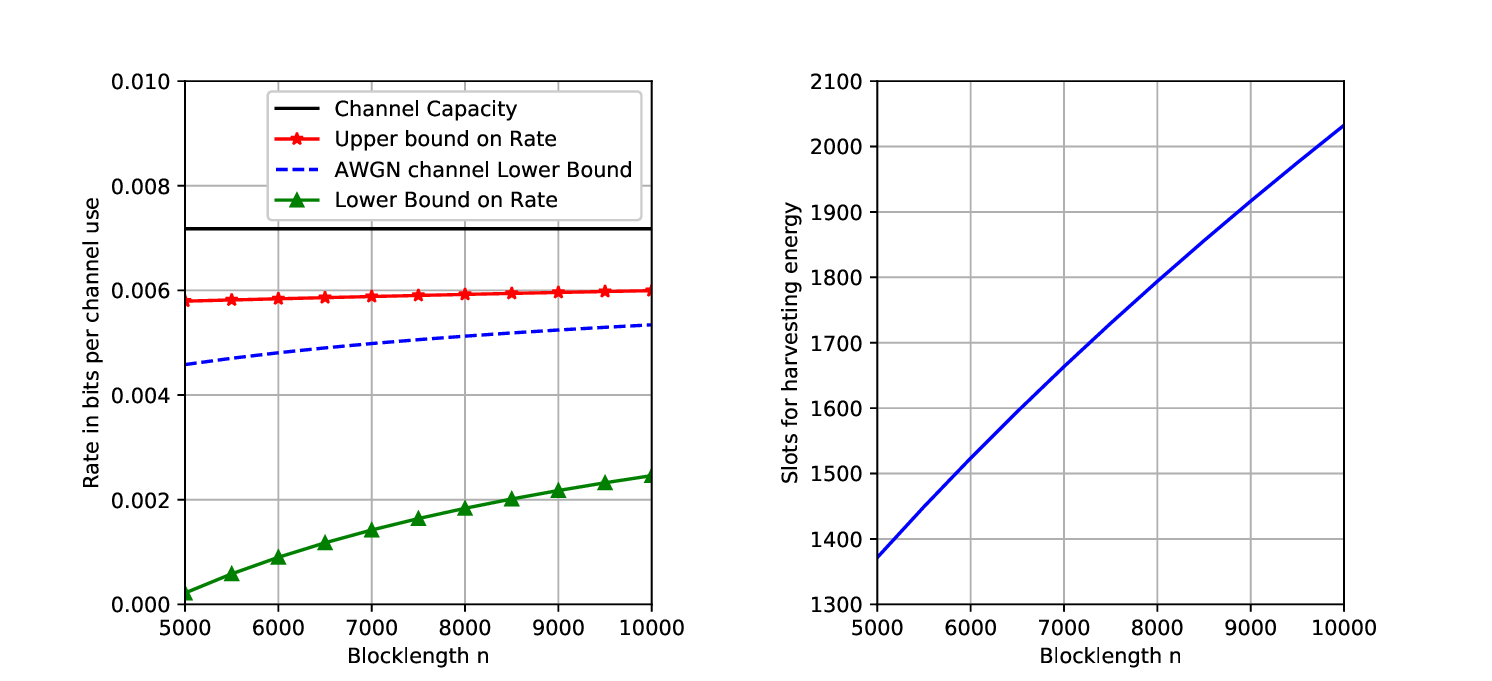}
		\caption{Plot of FB rates for an EH-AWGN channel versus the total blocklength (harvesting plus transmission) in low SNR regime. The other plot shows the number of slots used for harvesting energy.}\label{Plotb1}
	\end{figure}
	\item Moderate SNR ($0$ dB). Compared to low SNR, this regime (see Fig. \ref{Plotb2}) gives a better approximation to finite blocklength. The gap in rates is significantly lowered to approximately $19\%$ to $27\%$. Additionally, the number of slots required in the saving phase are also considerably reduced ($16\%$ to $22\%$).
	\begin{figure}
		\includegraphics[width=\columnwidth]{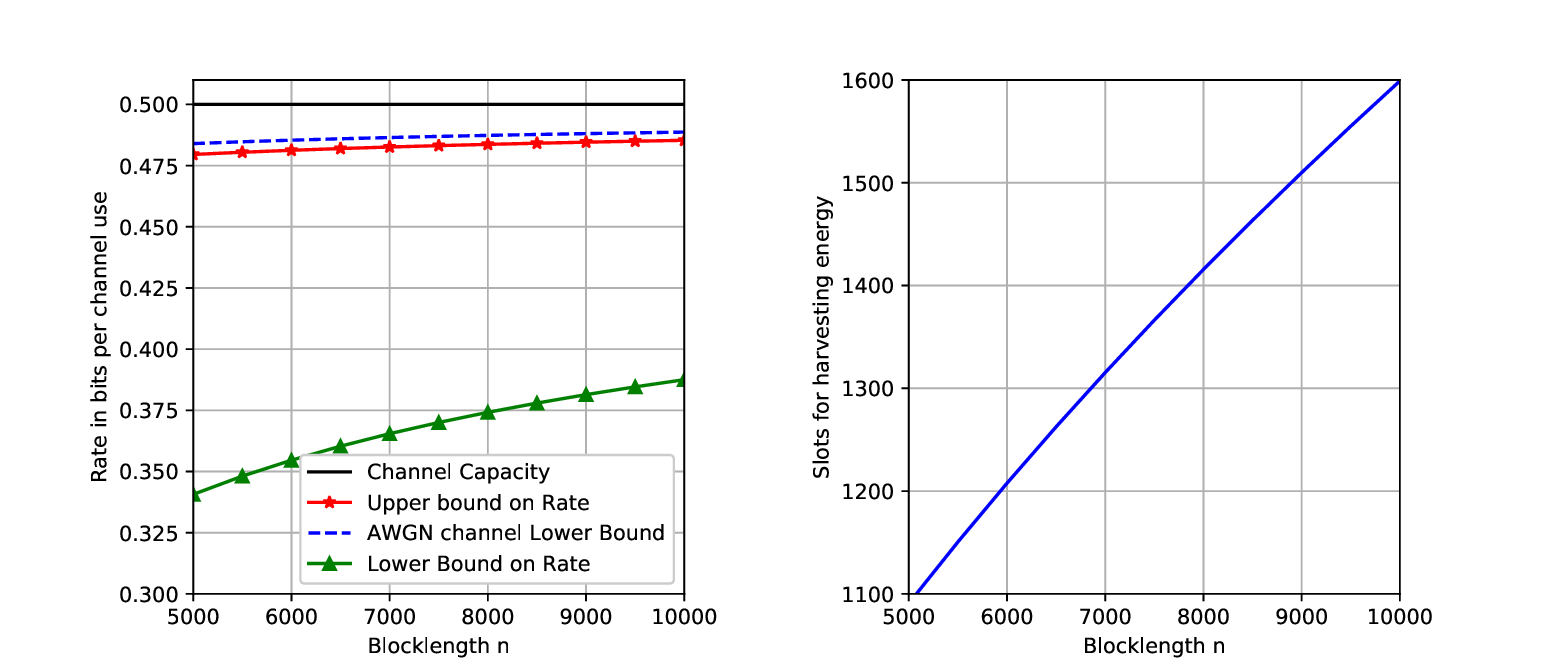}
		\caption{Plot of FB rates for an EH-AWGN channel versus the total blocklength (harvesting plus transmission) in moderate SNR regime. The other plot shows the number of slots used for harvesting energy.}\label{Plotb2}
	\end{figure} 
	\item High SNR ($20$ dB). In this regime (Fig. \ref{Plotb3}), the gap between rates is about $18.2\%$ to $24.2\%$ and the slots required in saving energy is between $15.8\%$ to $21.6\%$. While this is an improvement from moderate SNR, it is not as significant as that between low SNR to moderate SNR.
	\begin{figure}
		\includegraphics[width=\columnwidth]{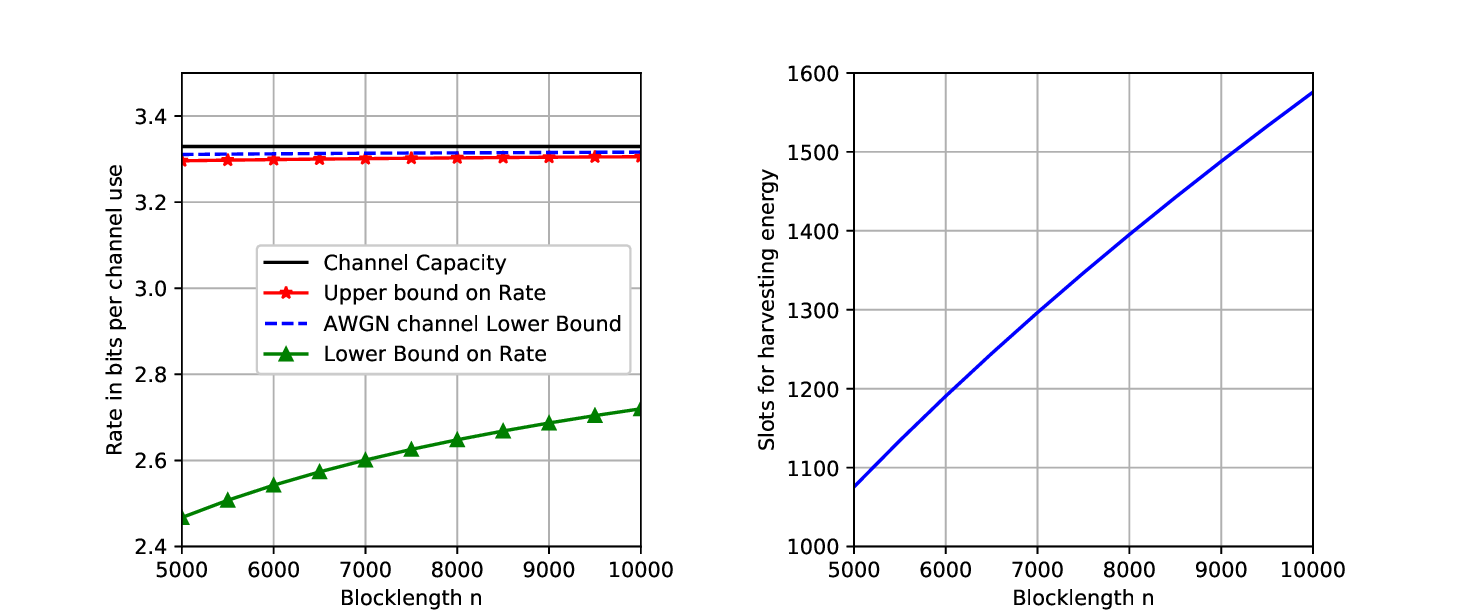}
		\caption{Plot of FB rates for an EH-AWGN channel versus the total blocklength (harvesting plus transmission) in high SNR regime. The other plot shows the number of slots used for harvesting energy.}\label{Plotb3}
	\end{figure}
\end{enumerate}
To summarize, the finite blocklength bounds are decent approximations to the finite blocklength rate in the moderate to high SNR regime. Further improvements would require an improved lower bound which would require changing the transmission scheme. Except for the low SNR case, we observe that the energy harvesting upper bound is below the lower bound of the equivalent AWGN channel. We can infer from this that the finite blocklength energy harvesting rates are lower than that of the non energy harvesting case in the moderate and high SNR regime. 

\subsection{EH-BSC}
Consider a binary symmetric channel W, with crossover probability $\alpha$. That is $\mathcal{X} =\mathcal{Y}= \{0,1\}$, $W(0|1) = W(1|0) = \alpha$. Let $p_0 := Pr(X = 0)$. If the capacity achieving distribution, which satisfies the energy harvesting requirements, is unique with $p_0$ as before, then
\begin{IEEEeqnarray}{rCl}
	C_{ED} = C_{BSC} &=& h(\alpha p_0 + \overline{\alpha}~\overline{p_0}) - h(\alpha),\notag\\
	V(P;W) = V_{BSC} &=& \sum_{x \in \{\alpha,\overline{\alpha}\}} \sum_{y \in \{p_0, \overline{p_0}\}} xy\left[ \log \left(\frac{x}{xy + \overline{x}~\overline{y}}\right)\right]^2 - C_{BSC}^2.\notag
\end{IEEEeqnarray}
where $\overline{u} := 1-u$ and $h(x) = -x\log_2(x)-\overline{x}\log_2(\overline{x})$ is the binary entropy function. Note that the choice of $p_0$ is influenced by energy harvesting constraints. In this example, we pick $\alpha = 0.05$, the energy function $\Lambda(x) = 3x$ and $\mathbb{E}[E_1] = 1$. This ensures the uniqueness of the capacity achieving distribution with $p_0 = 2/3$. We take $\varepsilon = 0.1$ and $\sigma_E^2 = 0.2$ here. Fig. \ref{Plotbbsc} plots the lower and upper bounds for this example where $n$ is between $5000$ and $10000$. 

\begin{figure}
	\includegraphics[width=\columnwidth]{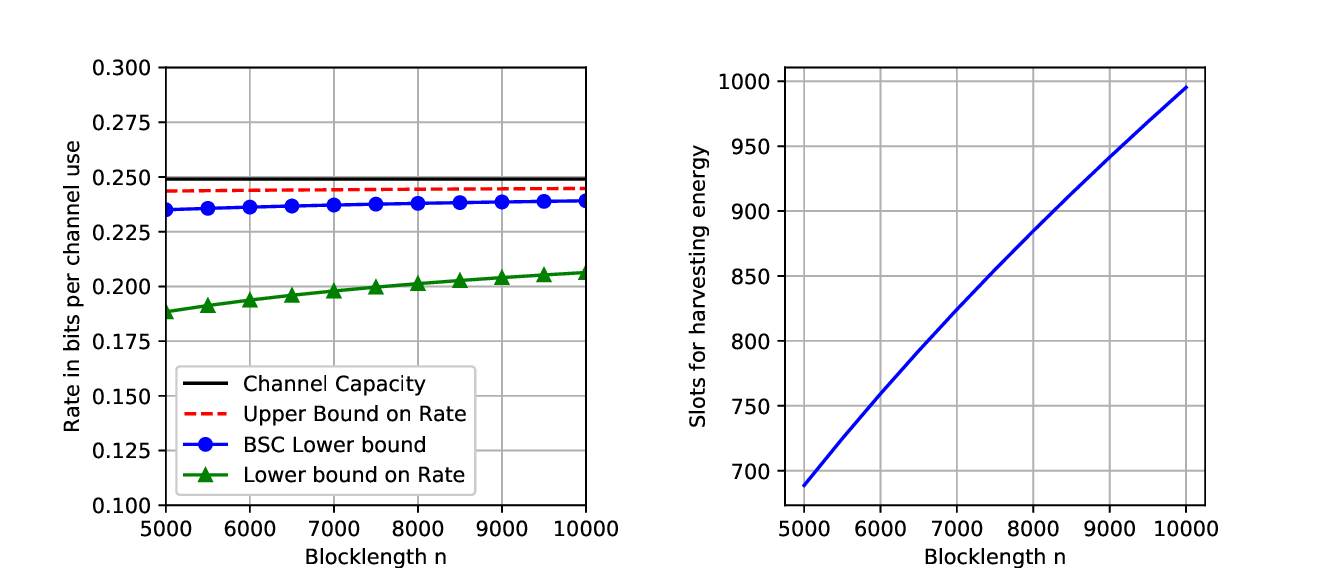}
	\caption{Plot of FB rates for an EH-BSC channel versus the total blocklength (harvesting plus transmission). The plot on the right gives the number of slots used for harvesting energy.}\label{Plotbbsc}
\end{figure}
 
 We observe that the difference between upper and lower bounds, for this example is between $13.7\%$ to $23\%$. The blocklength required for saving energy varies from $9.8\%$ to $13.8\%$ in this range. In this case, the non-energy harvesting lower bound is below the energy harvesting upper bound. Hence we cannot infer anything about the rates as a function of $\sigma_E^2$ here.
  
\subsection{EH-BEC}
A binary erasure channel $W$ is a channel with binary inputs $\mathcal{X} = \{0,1\}$, ternary outputs $\mathcal{Y} = \{0,E_R,1\}$ with $W(0|0) = W(1|1)= 1-\alpha$ and $W(E_R|0) = W(E_R|1) = \alpha$, where $\alpha$ is the erasure probability. Similar to the BSC case, if we have a unique capacity achieving distribution, $p_0 = Pr(X=0)$, then
\begin{IEEEeqnarray}{rCl}
	C_{ED} = C_{BEC} &=& (1 - \alpha)h(p_0),\notag\\
	V(P;W) = V_{BEC} &=& (1-\alpha)p_0(\log(p_0))^2 + (1-\alpha)(1-p_0)(\log(1-p_0))^2- C_{BEC}^2.\notag
\end{IEEEeqnarray}
Using the same parameters as in the BSC case, we plot the bounds in Fig. \ref{Plotbbec}.
\begin{figure}
	\includegraphics[width=\columnwidth]{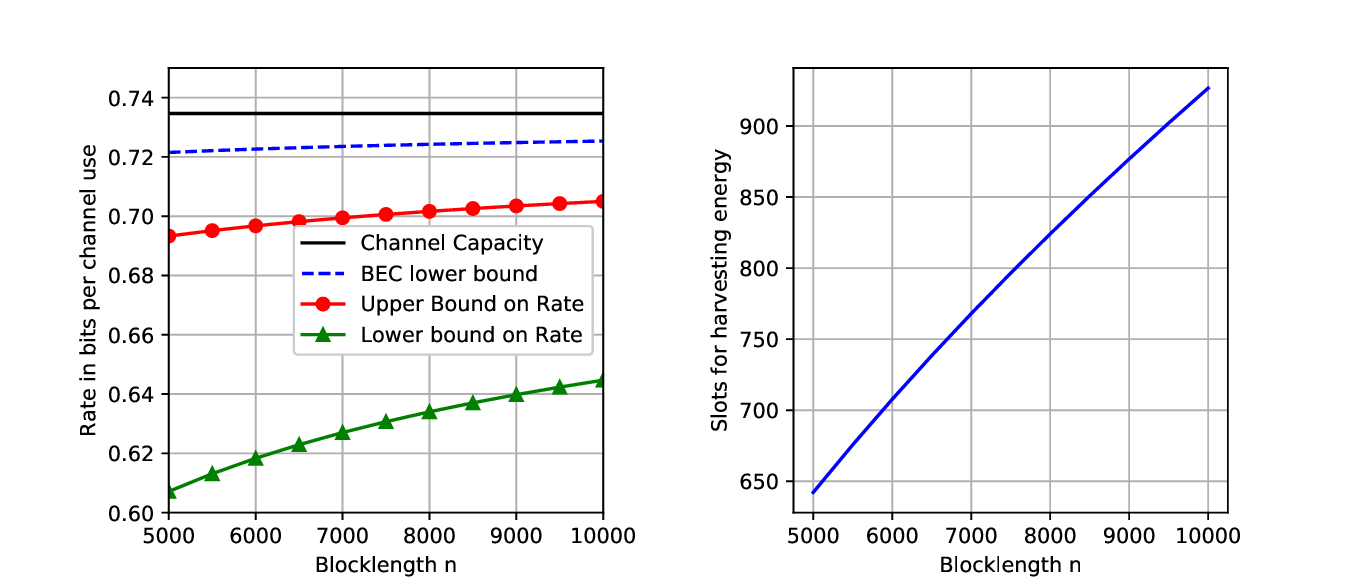}
	\caption{Plot of FB rates for an EH-BEC channel versus the total blocklength (harvesting plus transmission). The plot on the right gives the number of slots used for harvesting energy.}\label{Plotbbec}
\end{figure}

We observe a difference of $8.6\%$ to $12.2\%$ between the upper and lower bounds as well as saving energy slot utilization of $9.3\%$ to $12.8\%$ for the specified range of parameters. Here our bounds appear to better approximate the rates as opposed to BSC. Moreover, the non-energy harvesting lower bound is above the upper bound meaning that in this case, the effects of energy harvesting are detrimental to the rate.

\subsection{Effects of energy harvesting variance $\sigma_E^2$}
Comparing the bounds \eqref{ehawgnclb} and \eqref{ehawgncub} derived for EH-AWGN channel, we observe that both bounds are lowered with increasing $\sigma_E^2$. This is illustrated in Fig. \ref{Plotmany}.
\begin{figure}
	\includegraphics[width=\columnwidth]{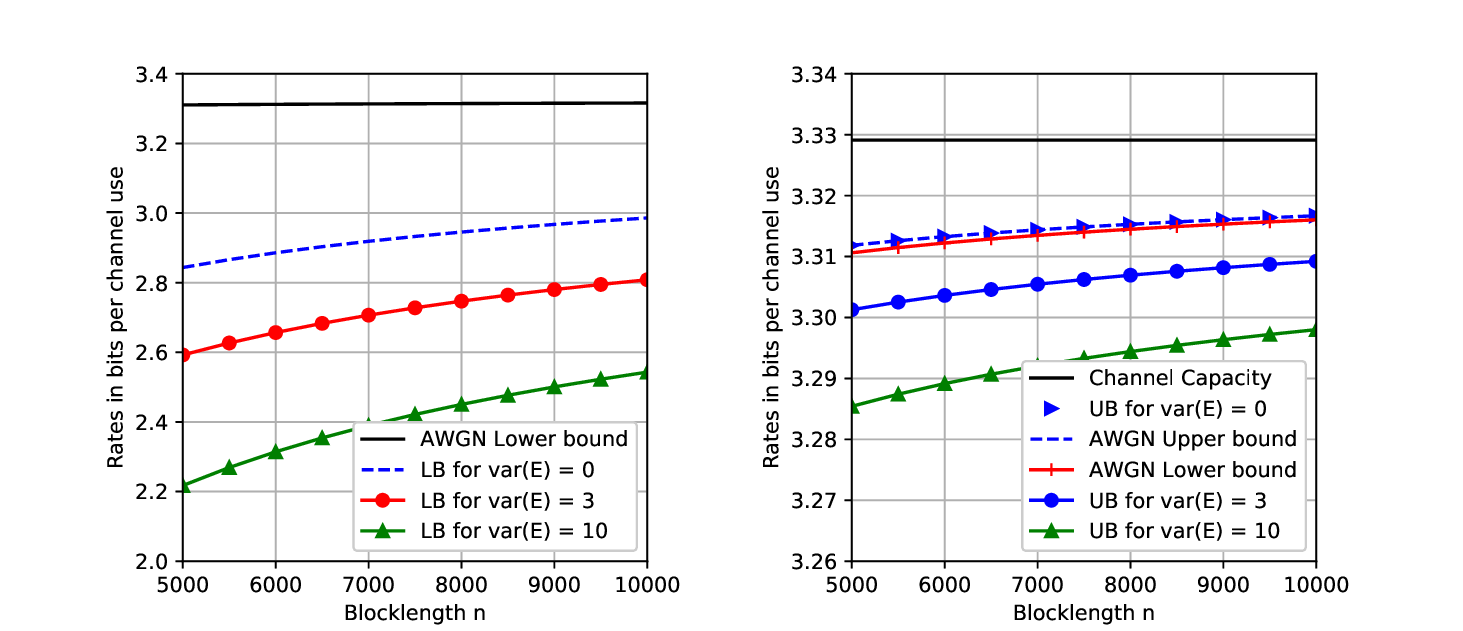}
	\caption{Plot of finite blocklength rates for an EH-AWGN for different energy harvesting variances.}\label{Plotmany}
\end{figure}
Interestingly, when compared to the AWGN lower bound, the EH-AWGN upper bound appears to only differ by $O(\log n/n)$ when $\sigma_E^2 = 0$. However the lower bound is strongly affected by the variance.

\section{Discussion of Results and Conclusion}
In this manuscript, we have shown that for both EH-AWGN and EH-DMC channels, the finite blocklength code size varies as $nC - \Theta(\sqrt{n})$ under the maximal probability of error criterion. This was shown by deriving lower and upper bounds with second order $\sqrt{n}$. We also bounded the moderate deviation asymptotics for both channel types. 

Additionally, the bounds were plotted for a few examples. In certain cases, such as the AWGN channel with moderate to high SNR as well as the BEC case, we observed that the rates are exacerbated with increased variance of the energy harvesting process. It is desirable to tighten the gap between lower and upper bounds so that this conjecture may be further verified.
%Based on the numerical results, it would indicate a substantial improvement would be possible by changing the achievable scheme. Unfortunately there doesn't appear to be a scheme as of now that can rival save and transmit in terms of finite blocklength bounds.
As future work, obtaining matching bounds in the finite blocklength as well as the moderate deviations regime will be useful.
\appendices
\section{Proof of (\ref{metaeh2}) }\label{stpr}
Let $U \in [M]$ denote the message to be transmitted and similarly $\hat{U}$ the decoded message. We have for channel $W$, if the maximal probability of error is $\varepsilon$, the following steps.
\begin{IEEEeqnarray}{rCl}
1-\varepsilon &\leq& Pr[\hat{U}=m|U=m)\notag \\
&=& \int_{\mathbf{y},\mathbf{e}}Pr[\hat{U}=m|\mathbf{Y}=\mathbf{y}]W(\mathbf{y}|c(m,\mathbf{e}))dP_\mathbf{E}(\mathbf{e}) \label{bet11}
\end{IEEEeqnarray}
where the above holds for any message $m$.

Now $Pr[\hat{U}=m|\mathbf{Y}=\mathbf{y}]$ is a test on the decoder end that achieves the probability of error requirement. Even though it doesn't depend on $\mathbf{e}$, since the decoder doesn't have access to the energy samples, it is still a valid test on $(\mathbf{y},\mathbf{e})$. 

Now suppose instead of channel $W$, the message is sent on channel $Q_Y$ which is an auxiliary channel that ignores the input but has the same output alphabet. Using the above decoder, let $\overline{m}$ be the message that achieves the maximal probability of error under $Q_Y$. Then clearly $P(\hat{U}=\overline{m}|U=\overline{m}) \leq \frac{1}{M}$ under $Q_Y$. But then, from (\ref{bet11}) and the definition of the beta error function, we have
\begin{equation}
\beta_{1-\varepsilon} \left(W(.|c(\overline{m},*))P_{\mathbf{E}}(*),Q_{\mathbf{Y}}P_{\mathbf{E}}\right) \leq \int_{\mathbf{y}}P(\hat{U}=\overline{m}|\mathbf{Y}=\mathbf{y} )Q_{\mathbf{Y}}(\mathbf{y}) \leq \frac{1}{M}\notag
%\label{finb11}
\end{equation}

\section{Proof of Theorem \ref{Thuse}} \label{Thusea}
The proof follows the steps used in proving the original meta-converse (see \cite{PPV1}) upto a point. Given distribution $Q_{\mathbf{Y}}$, which is essentially a reference channel that does not depend on input, let the maximal probability of error for this ``channel'' be $\varepsilon'$. Let $U$ be the random variable denoting the message to be sent and $\hat{U}$ be the message that was decoded.

Consider the definition of maximal probability of error. We see that there is a message, call it $\overline{m}$ such that
\begin{equation}
1 - \varepsilon' = Pr\left[\hat{U} = \overline{m}| U = \overline{m}\right] = \int\limits_\mathbf{y} P_{\hat{U}|\mathbf{Y}}(\overline{m}|\mathbf{y})dQ_{\mathbf{Y}}(\mathbf{y}).
\label{maxp}
\end{equation}
But we also have 
\begin{IEEEeqnarray}{rCl}
1 - \varepsilon' &=& \min \limits_m Pr\left[\hat{U} = m| U = m\right]\notag \\
&\leq& \frac{1}{M} \sum_{m=1}^M Pr\left[\hat{U} = m| U = m\right]\notag\\
&=& \frac{1}{M} \sum_{m=1}^M \int \limits_{\mathbf{y}} Pr\left[\hat{U} = m| \mathbf{Y} = \mathbf{y}\right]dQ_{\mathbf{Y}}(\mathbf{y})\notag \\
&=& \frac{1}{M} \int \limits_{\mathbf{y}} \left( \sum_{m=1}^M Pr\left[\hat{U} = m| \mathbf{Y} =\mathbf{y}\right] \right) dQ_{\mathbf{Y}}(\mathbf{y})\notag \\
&=&\frac{1}{M}
\label{maxp2}
\end{IEEEeqnarray}
Combining equation (\ref{maxp}) and (\ref{maxp2}), we get
\begin{equation}
M \leq \frac{1}{\int\limits_\mathbf{y} P_{\hat{U}|\mathbf{Y}}(\overline{m}|\mathbf{y})dQ_{\mathbf{Y}}(\mathbf{y})}.
\label{maxp3}
\end{equation}

%First we see that there is a message $\overline{m}$ for which
%\begin{equation}
%\frac{1}{M} \geq 1-\varepsilon' = \int\limits_y P_{\hat{U}|Y^n}(\overline{m}|y^n)dQ_{\mathbf{Y}}(y^n).
%\end{equation}
%where $\varepsilon'$ is the maximal probability of error if we had $Q_{\mathbf{Y}}$ as the channel. 

Now we have for any $\mathcal{E}_1\subset \mathbb{R}_+^n$,
\begin{IEEEeqnarray}{rCl}
1-\varepsilon &\leq& \int\limits_{\mathbf{e}} \int\limits_{\mathbf{y}} P_{\hat{U}|\mathbf{Y}}(\overline{m}|\mathbf{y})dP_{\mathbf{Y}|\mathbf{X}}(\mathbf{y}|c(\overline{m},\mathbf{e}))dP_\mathbf{E}(\mathbf{e})\notag \\
&\leq& \int \limits_{\mathbf{e}\in \mathcal{E}_1} \int\limits_{\mathbf{y}} P_{\hat{U}|\mathbf{Y}^n}(\overline{m}|\mathbf{y})dP_{\mathbf{Y}|\mathbf{X}}(\mathbf{y}|c(\overline{m},\mathbf{e}))dP_\mathbf{E}(\mathbf{e}) +P_\mathbf{E}(\mathcal{E}_1^c).\notag
\end{IEEEeqnarray}
Rearranging and using the definitions given in the statement of the lemma, letting $\mathcal{E}_1 = \{\mathbf{e}: \sum_{i=1}^n e_i \leq n\overline{E}_n\}$ and $\tau_n = P_E(\mathcal{E}_1^c)$, we get
\begin{IEEEeqnarray}{rCl}
1-\varepsilon-\tau_n &\leq  \int \limits_{\mathbf{e}\in \mathcal{E}_1} \int\limits_{\mathbf{y}}\notag P_{\hat{U}|\mathbf{Y}}(\overline{m}|\mathbf{y})dP_{\mathbf{Y}|\mathbf{X}}(\mathbf{y}|c(\overline{m},\mathbf{e}))dP_{\mathbf{E}}(\mathbf{e}) \\
\Rightarrow 1-\varepsilon-\tau_n \leq\frac{1-\varepsilon-\tau_n}{1-\tau_n} &\leq  \int\limits_{\mathbf{y}} P_{\hat{U}|\mathbf{Y}}(\overline{m}|\mathbf{y})  \left\{ \int\limits_{\mathbf{e}\in \mathcal{E}_1}dP_{\mathbf{Y}|\mathbf{X}}(\mathbf{y}|c(\overline{m},\mathbf{e}))\frac{dP_{\mathbf{E}}(\mathbf{e})}{1-\tau_n}\right\}. \label{key2}
\end{IEEEeqnarray}
Note that we divide by $1-\tau_n$ is to ensure that the term in braces is a probability distribution. From (\ref{maxp3}), (\ref{key2}) and the definition of $\beta$ error function, we get
%\begin{equation}
%\frac{1-\varepsilon-\tau_n}{1-\tau_n} = 1-\varepsilon-\frac{\tau_n\varepsilon}{1-\tau_n} \geq1-a\varepsilon,
%\label{interm}
%\end{equation} 
%where the last inequality follows from using Lemma \ref{maukeq}. 
\begin{IEEEeqnarray}{rCl}
\frac{1}{M} &\geq& \beta_{1-\varepsilon-\tau_n}\left(\int_{\mathbf{e}\in \mathcal{E}_1}dP_{\mathbf{Y}|\mathbf{X}}(.|c(\overline{m},\mathbf{e}))\frac{dP_{\mathbf{E}}(\mathbf{e})}{1-\tau_n},Q_{\mathbf{Y}} \right) \notag\\
&\geq& \inf \limits_{\mathbf{x} \in\mathbb{F}_{\overline{E}_n}} \beta_{1-\varepsilon-\tau_n}\left(\int_{\mathbf{e}\in \mathcal{E}_1}dP_{\mathbf{Y}|\mathbf{X}}(.|\mathbf{x})\frac{dP_{\mathbf{E}}(\mathbf{e})}{1-\tau_n},Q_{\mathbf{Y}} \right) \notag\\
&=& \inf \limits_{\mathbf{x} \in\mathbb{F}_{\overline{E}_n}} \beta_{1-\varepsilon-\tau_n}\left(P_{\mathbf{Y}|\mathbf{X}}(.|\mathbf{x}),Q_{\mathbf{Y}} \right).\notag
\label{conv2}
\end{IEEEeqnarray}
Note that we could take the infimum over $\mathbb{F}_{\overline{E}_n}$, a non-random set here because when $e^n \in \mathcal{E}_1$, it implies that $c(\overline{m},e^n) \in \mathbb{F}$. Hence we have (\ref{metaeh3}).

\bibliographystyle{IEEEtran}
\bibliography{journbib}

\end{document}